\documentclass[11pt,dvipsnames]{article}
\usepackage{ssArxiv}

\usepackage[utf8]{inputenc} 
\usepackage[T1]{fontenc}    
\usepackage{microtype}      
\usepackage[rgb]{xcolor}
\usepackage{caption}
\usepackage{cite}
\usepackage{textcomp}
\def\BibTeX{{\rm B\kern-.05em{\sc i\kern-.025em b}\kern-.08em
    T\kern-.1667em\lower.7ex\hbox{E}\kern-.125emX}}

\usepackage{amsmath} 
\usepackage{amsthm}
\usepackage{amssymb}  
\usepackage{amsfonts}       
\usepackage{nicefrac}       
\usepackage{accents}

\DeclareMathOperator{\relint}{relint}
\DeclareMathOperator{\interior}{int}
\DeclareMathOperator{\range}{rg}
\DeclareMathOperator{\aff}{aff}

\newcommand{\SetA}{\mathcal{A}}
\newcommand{\SetAbar}{\bar{\SetA}}
\newcommand{\SetB}{\mathcal{B}}
\newcommand{\Ball}{\mathcal{B}_{P}}
\newcommand{\SetP}[1]{\mathcal{P}(#1)}
\newcommand{\SetW}{\mathcal{W}}
\newcommand{\SetX}{\mathcal{X}}
\newcommand{\SetY}{\mathcal{Y}}
\newcommand{\SetZ}{\mathcal{Z}}
\newcommand{\SetV}{\mathcal{V}}
\newcommand{\SetVaug}{\mathcal{V}_{a}}
\newcommand{\Reals}[1]{\mathbb{R}^{#1}}
\newcommand{\SetS}{\mathcal{S}}
\newcommand{\map}{M}
\DeclareMathOperator{\Prob}{Pr}
\DeclareMathOperator{\bd}{bd}

\newtheorem{lemma}{Lemma}
\newtheorem{corollary}{Corollary}
\newtheorem{remark}{Remark}
\newtheorem{proposition}{Proposition}
\newtheorem{definition}{Definition}
\newtheorem{assumption}{Assumption}
\newtheorem{example}{Example}

\usepackage{graphicx} 
\usepackage{wrapfig}
\usepackage{subcaption}
\usepackage{multirow}
\usepackage{array}
\usepackage{booktabs}      
\usepackage{stfloats}
\usepackage{tikz}

\usepackage{algorithm}
\usepackage{algpseudocode} 
\usepackage{verbatim}

\usepackage{soul}
\usepackage{siunitx}
\definecolor{dodgerblue}{RGB}{30,144,255}
\definecolor{tomato}{RGB}{255,99,71}
\definecolor{seagreen}{RGB}{46,139,87}
\definecolor{orange}{RGB}{255,165,0}
\definecolor{orchid}{RGB}{218,112,214}

\setstcolor{red}
\newcommand{\legendline}[2]{\textcolor{#1}{\rule{6pt}{2pt}}\ #2}

\usepackage{url}            
\usepackage{hyperref}       
\usepackage{cleveref}       
\crefname{figure}{Fig.}{Figs.} 
\crefname{assumption}{Assumption}{Assumptions}

\title{Bridging Reinforcement Learning and Optimal Control via Feasible Action Mapping}
\author{\name{Stefan Richter$^{*}$}
\email{sr@richteroptimization.com} \\
\addr{Richter Optimization GmbH, Vienna, Austria},  \\[0.5em]
\name{Alberto Giammarino$^{*}$}
\email{alberto.a.giammarino@sony.com} \\
\addr{Sony AI, Tokyo, Japan} \\[0.5em]
\name{Guillem Torrente}
\email{guillem.torrente@sony.com} \\
\addr{Sony AI, Tokyo, Japan} \\[0.5em]
\name{Sam Blakeman}
\email{samrobertallan.blakeman@sony.com} \\
\addr{Sony AI, Zurich, Switzerland} \\[0.5em]
\name{Peter Dürr} 
\email{peter.duerr@sony.com} \\
\addr{Sony AI, Zurich, Switzerland} \\[0.5em]
}

\begin{document}
\maketitle
\small{$^{*}$ Contributed equally to this work.}

\begin{abstract}
Operating constrained dynamical systems requires controllers to efficiently solve complex tasks while enforcing recursive feasibility and physical constraints. 
To address these competing requirements, we present Feasible Action for Optimal Control (FAOC), a novel control framework integrating Reinforcement Learning (RL) and Optimal Control (OC). 
The core contribution is a computationally efficient, optimization-based mapping algorithm that transforms the RL agent's action from a static abstract set into a state-dependent feasible parameter set of the Optimal Control Problem (OCP), guaranteeing instantaneous parameter feasibility. 
When paired with invariant terminal sets, FAOC guarantees strict recursive feasibility and safe operation, effectively combining the predictable safety of OC with the behavioral flexibility of RL.
Unlike prior work, the abstract action space does not require expert tuning, nor is the OCP formulation compromised by the inability of RL to guarantee feasibility. 
We evaluate FAOC on real-time motion planning for robot table tennis, where simulated experiments demonstrate superior sample efficiency and closed-loop performance compared to state-of-the-art baselines.
We open-source the used implementation of the mapping algorithm and OCP for motion planning~\footnote{\url{https://github.com/SonyResearch/feasible_action_for_optimal_control}}.
\end{abstract}


\section{Introduction}
\label{sec:Introduction}
\noindent

Recent technological breakthroughs have significantly advanced the state of the art in the control of constrained dynamical systems.
One clear example is robotics, where fully autonomous robots, operating without pre-programmed routines, are now capable of performing complex tasks at human-level proficiency \cite{ace2026, kaufmann2012champion, luo2025preciseanddexterous}.
Reinforcement Learning (RL) has been a cornerstone of this progress, demonstrating the capability to solve complex control problems across several domains \cite{kaufmann2012champion, luo2025preciseanddexterous, wurman2022,hoeller2024anymal, hwangbo2019learning}.
Pure RL-based controllers, however, exhibit significant shortcomings, most notably the inability to strictly guarantee constraint satisfaction.
These limitations have restricted the adoption of RL in real-world applications, which typically require strict safety guarantees~\cite{robotics13040063, pmlr-v267-fan25i, GROS20208076}.

Another fundamental approach to the closed-loop control of constrained dynamical systems is Optimal Control (OC), specifically Model Predictive Control (MPC), which offers complementary properties to RL.
Most notably, OC incorporates prior knowledge of the system dynamics and explicitly enforces safety through input, state, and output constraints, ensuring safe operation even in unseen environments.
Nevertheless, the scale and complexity of formulating these tasks as Optimal Control Problems (OCPs) increase solution times, sometimes rendering the computation intractable within strict real-time deadlines.
Furthermore, OC remains challenging to apply to nonconvex, long-horizon problems with sparse reward signals, limiting its effectiveness in tasks requiring extensive exploration or delayed credit assignment \cite{torrente2021data, hewing2020learning, GARCIA1989335, MAYNE2000789, Borrelli_Bemporad_Morari_2017}.

The idea of combining RL and OC to leverage their complementary strengths has been widely explored (see~\cite{Reiter2025SynthesisOM} for a comprehensive survey and  \Cref{sec:literature_review} for a summary of related work). To explicitly handle long-horizon problems and sparse reward signals while retaining the advantages of OC, one common approach is to learn the terminal cost of the OCP using a Neural Network (NN) that approximates the infinite-horizon cost-to-go~\cite{zhong2013value, reiter2025ac4mpc, lowrey2018plan}. However, this approach generally renders the OCP nonconvex, which can severely impact real-time solvability.

Another common formulation combines RL with a parameterized OCP, where the RL agent selects OCP parameters to maximize long-term cumulative rewards. While successful, this coupling is difficult when RL actions affect OCP feasibility: some actions may render the OCP infeasible, requiring slack variables or careful reformulation. For tractability, existing methods often employ simple, static, continuous RL action spaces; however, the actual feasible parameter sets of the underlying OCP are typically state-dependent and geometrically complex. Consequently, static parameterizations may include infeasible or redundant actions while excluding optimal feasible ones. This forces the RL policy to implicitly learn poorly defined feasibility boundaries, thereby degrading learning efficiency and final performance~\cite{zhu2022overview}.

\begin{figure*}[t]
    \centering
    \includegraphics[trim=0.0cm 11.7cm 0.8cm 0.0cm,clip,width=\linewidth]{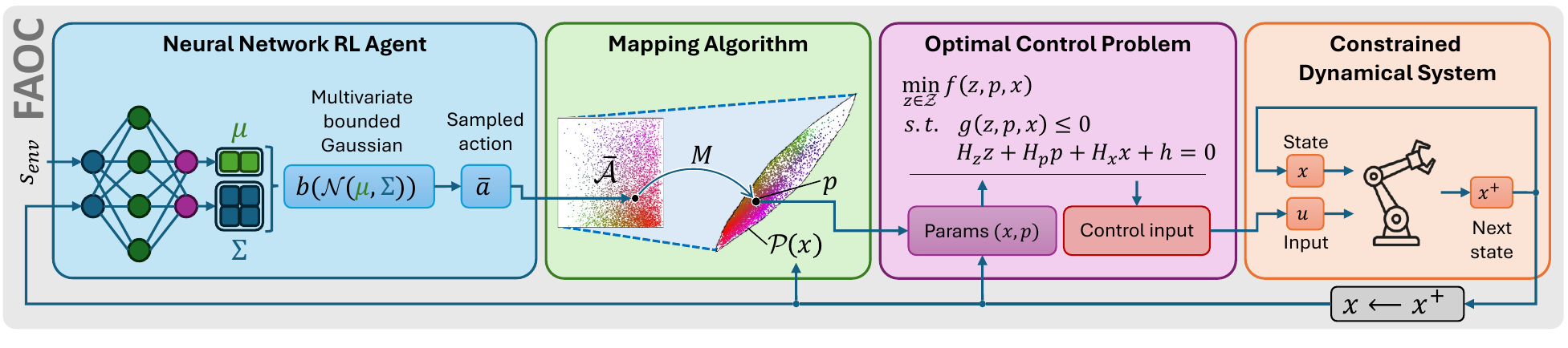}
    \caption{
        Overview of the FAOC framework illustrated for a 2D action space. From left to right:
        \emph{Step 1:} An RL agent outputs the parameters (mean~$\mu$ and covariance~$\Sigma$) of a multivariate Gaussian distribution, given the system state $x$ and potentially other observations $s_{env}$. A raw action is sampled from this distribution during training, or taken as the mean during deployment. A nonlinear function~$b$ transforms this unbounded raw action into an abstract action~$\bar{a}$ confined within a static, compact abstract set~$\SetAbar$.
        \emph{Step 2:} The abstract action~$\bar{a}$ is mapped via $\map$ into a parameter set~$\SetP{x}$ that depends on the current state~$x$ of the dynamical system. This set encompasses all feasible parameters for the underlying OCP (e.g., terminal state targets). The color gradient visualizes the continuous mapping between the sets.
        \emph{Step 3:} The mapped parameter~$p$ fixes the action-encoding vector in the OCP, which is subsequently solved to yield an optimal control input trajectory.
        \emph{Step 4:} This trajectory is applied to the system in open loop. The resulting subsequent state~$x^+$ is observed and serves as the initial condition for the next control cycle.
    }
    \label{fig:intro}
\end{figure*}

\subsection*{Main Contributions}
We introduce \emph{Feasible Action for Optimal Control} (FAOC), a novel RL-OC framework for constrained dynamical systems resolving persistent feasibility and exploration challenges. FAOC enables the RL policy to sample from a static, geometrically simple abstract set while guaranteeing the resulting OCP remains feasible and all valid parameters remain reachable. This relieves the RL agent from satisfying physical constraints, focusing its capacity solely on the overarching task (\Cref{fig:intro}).

Our main contributions are:%
\begin{itemize}
    \item[(i)] \textbf{Topological Characterization of OCP Feasibility.} We identify mild geometric conditions ensuring the state-dependent parameter set $\SetP{x}$ of a generic OCP is compact, solid, and convex (\Cref{lem:compactness_of_P}), explicitly demonstrated for linear MPC (\Cref{cor:linear_mpc}).
    \item[(ii)] \textbf{Invertible Feasible Action Mapping.} We propose a computationally efficient radial algorithm (\Cref{alg:basic}) bijectively transforming any abstract RL action $\bar{a} \in \SetAbar$ into a guaranteed-feasible parameter $p \in \SetP{x}$ (\Cref{prop:Finv_Alg1}).
    \item[(iii)] \textbf{Mitigation of Geometric Distortion.} To prevent point accumulation, we develop an area-matching directional transformation (\Cref{prop:phi_properties,prop:marginal}), and a scalable linear transformation utilizing affinely related ellipsoidal surrogates for arbitrary dimensions (\Cref{prop:ellipsoid_surrogate}).
    \item[(iv)] \textbf{Tractability for Implicitly Defined Sets.} For real-time scalability without explicit geometric representations of $\SetP{x}$, we derive robust formulations for computing interior points (\Cref{prop:vertex_inscribed,prop:dual_inscribed}) and target shape matrices (\Cref{prop:implicit_surrogate}) directly from OCP constraints (\Cref{ex:constant_phi}).
    \item[(v)] \textbf{Validation in Robot Table Tennis.} We apply FAOC to motion planning for an 8-DoF robot~arm (\Cref{sec:exp_setup}). In simulation, FAOC consistently surpasses state-of-the-art RL-MPC controllers in sample efficiency, final performance, and system~controllability. Notably, while the mapping algorithm assumes convex sets, this application highlights FAOC's capability in nonconvex domains: the RL agent resolves strategic nonconvexities (e.g., obstacle avoidance, nonlinear kinematics), while the OCP handles local kinematic constraints yielding a convex set amenable to mapping. Our mapping algorithm and OCP implementation are open-sourced~\footnote{\url{https://github.com/SonyResearch/feasible_action_for_optimal_control}}.
\end{itemize}

\section{Related work}
\label{sec:literature_review}
\noindent
Both RL and OC enable state-feedback control of constrained dynamical systems, yet they exhibit complementary strengths and weaknesses.
In particular, RL typically suffers from poor sample efficiency and lacks guarantees of constraint satisfaction.
Conversely, while OC systematically enforces constraints, its application to complex, nonconvex problems often suffers from computational intractability and sensitivity to initialization.
Specifically, local optimization solvers may converge to infeasible stationary points even when a globally feasible solution exists.
Consequently, combining these two technologies to leverage their respective advantages has emerged as a promising research direction.

Recent literature extensively explores this intersection, particularly focusing on MPC (see~\cite{Reiter2025SynthesisOM} for a comprehensive review).
For context, we first situate our contribution within the classification described in~\cite{Reiter2025SynthesisOM}.
Our approach belongs to the category of works in which a parameterized MPC is treated as part of the RL environment, and the RL agent learns state-dependent values for its parameters.
Furthermore, we consider a hierarchical architecture in which the RL policy does not depend on the OCP decision variables.
As a result, the RL actions---namely, the parameters of the MPC framework---are treated as constants rather than decision variables within the underlying OCP.

Most existing works adopt parameterizations that do not influence OCP feasibility, such as cost function parameterizations, since ensuring feasibility is generally beyond the capabilities of an RL agent.
While this allows any RL action to be admissible, it raises the fundamental question of how to define meaningful action space bounds.
In such formulations, the action space becomes a hyperparameter that must be tuned.
For instance, in~\cite{brito2021go}, the RL agent learns subgoals for a navigation task, with actions defined in a 2D space bounded heuristically by the maximum reachable distance (computed by applying maximum velocity over the horizon).
To guarantee OCP feasibility even if the desired goal is unreachable, the distance to the goal is penalized in the terminal cost rather than enforced as a hard constraint.
Similar cost penalization strategies are used in~\cite{brito2022learning} for tracking velocity references in dense traffic scenarios.
In~\cite{reiter2023hierarchical}, RL is combined with nonlinear MPC where actions corresponding to reference lateral position and speed in Frenet coordinates~\cite{reiter2021parameterization} are mapped via a domain-informed expansion function to a high-dimensional MPC cost parameterization.
Other examples include quadrotor waypoint selection~\cite{greatwood2019reinforcement} and learning parameters for both the cost and control barrier functions using an unbounded action space~\cite{sabouni2024reinforcement}.
An exception to these heuristic bounds is~\cite{zarrouki2024adaptive}, where the RL agent learns two parameters (the stochastic propagation horizon and the constraint robustification factor) that are naturally bounded by physical and probabilistic considerations.
However, such well-bounded parameterizations are problem-specific and not easily generalized.

In approaches where OCP feasibility directly depends on the RL-selected parameters, but the action space itself does not inherently guarantee feasibility, the RL agent must explicitly learn to avoid infeasible actions, which ultimately hampers final performance.
For example, in~\cite{tram2019learning}, RL actions define constraints for intersection handling, and infeasible configurations are penalized through the reward function.
Similarly, in~\cite{pfrommer2022safe}, a chance-constrained MPC acts as a safety filter during training; infeasible OCPs trigger constraint relaxation and re-solving until the agent eventually learns safe actions.
An alternative approach is presented in~\cite{zarrouki2024safe}, where the RL action space is restricted to a finite set of Pareto-optimal MPC weight configurations obtained via Bayesian Optimization.
While this guarantees safety and enables efficient exploration, the discretization heavily restricts the action space and limits overall performance.

In this work, we simultaneously address the two primary issues highlighted above: (i) mitigating the need for ill-defined, tuned action spaces that induce trade-offs, and (ii) ensuring OCP feasibility when RL actions directly define constraints.
We achieve this by mapping the RL agent's output from a static abstract action set into a feasible, state-dependent parameter set of the OCP, ensuring that all feasible actions are included and infeasible ones are strictly excluded.

A related idea was proposed in~\cite{kiemel2021learning}, where a one-dimensional linear mapping transforms the RL output from a static range~$[-1,1]$ into a feasible set determined by the current system state.
While conceptually similar, this approach is limited to independent one-dimensional mappings.
As we demonstrate in our numerical application (\Cref{sec:results}), extending this concept to mappings over general $n$-dimensional sets provides greater expressiveness and significantly improves performance.

Another attempt to generalize this mapping is presented in~\cite{Yuan2023Action}; however, their method assumes the abstract action set can be represented as a hypercube and requires the feasible parameter set to contain the origin---both of which are restrictive assumptions.
In contrast, our method supports arbitrary compact, convex sets by incorporating a robust procedure to compute an interior point, thereby eliminating the restrictive assumption that the origin must lie within the interior of the feasible parameter set.
Additionally, it accounts for the geometric shapes of both the abstract set and the state-dependent parameter set, mitigating point accumulation and thus accelerating learning.

Finally, we note that \cite{theile2025action} proposes learning the mapping via supervised learning to handle nonconvex or disconnected feasible sets, whereas our exact analytical approach focuses on strict constraint guarantees.


\section{Problem Setup and Assumptions}
\noindent
Consider a constrained dynamical system operating in an environment that requires complex, potentially long-horizon decision making. To govern such a system, we consider a hierarchical control architecture: a high-level RL policy evaluates the state of the environment to select strategic behaviors, while a low-level parameterized OCP computes optimal control inputs that satisfy physical and safety constraints while tracking these high-level targets.

The fundamental challenge within this architecture arises at the interface between the two layers. We assume that the RL policy operates in a static, geometrically simple abstract space. In contrast, the underlying OCP is governed by strict physical limits, which induce a complex, state-dependent feasible region for the action-encoding parameter vector. If the RL agent selects an abstract action that maps to an incompatible parameter configuration, the OCP becomes infeasible, causing the control system to fail. Therefore, our goal is to design an intermediary mapping algorithm that safely and bijectively translates the RL agent's abstract action into a guaranteed-feasible OCP parameter.

To ensure that the proposed FAOC framework encompasses a broad class of systems, we first establish the general topological properties required of both the RL policy output and the parameterized OCP. While we keep these initial structural assumptions as mild as possible to maximize theoretical generizability, we introduce specific computational restrictions (such as polytopic representations) in \Cref{sec:mapping} to facilitate a tractable, real-time algorithmic implementation.

\subsection{Reinforcement Learning Interface}
\label{sec:preliminaries_RL}

A high-level sequential decision-making process can be modeled as a finite-horizon discounted Markov decision process (MDP) described by the tuple $(\SetS, \SetA, \kappa, r, \rho_0, \gamma)$. Here, $\SetS$ is the continuous state space, $\SetA \subseteq \Reals{n}$ is the action space, $\kappa:\SetS\times\SetA \rightarrow P(\SetS)$ represents the stochastic transition kernel (where $P(\SetS)$ denotes the space of probability distributions over $\SetS$), $r:\SetS\times\SetA\rightarrow \Reals{}$ is the reward function, $\rho_0 \in P(\SetS)$ is the initial state distribution, and $\gamma \in [0, 1)$ is the discount factor.
We focus on continuous control problems where the learning agent is modeled as a parameterized stochastic policy $\pi_{\theta}(a|s)$, parametrized by a deep neural network with weights~$\theta$. The deep RL objective is to learn the optimal weights~$\theta^*$ that maximize the expected total discounted reward over a horizon~$T$:
\begin{equation}
    \label{eq:RL_obj}
    J(\theta) = \mathop{\mathbb{E}}_{\substack{s_0 \sim \rho_0 \\ a_k \sim \pi_{\theta}(\cdot|s_k) \\ s_{k+1} \sim \kappa(\cdot | s_k, a_k)} }\left[\sum_{k=0}^T \gamma^k r(s_k, a_k) \right].
\end{equation}

Because neither the transition dynamics $\kappa$ nor the reward function $r$ are available to the learning agent, interaction with the environment is required.
A standard method to maximize the objective \eqref{eq:RL_obj} relies on estimating the action-value function $Q^{\pi}(s,a)$, which represents the expected cumulative discounted reward of taking action $a$ in state $s$ and subsequently following policy $\pi$:
\begin{equation*}
    Q^{\pi}(s,a)=\!\!\!\mathop{\mathbb{E}}_{\substack{a_k \sim \pi(\cdot|s_k) \\ s_{k+1} \sim \kappa(\cdot|s_k,a_k)}}\left[\sum_{k=0}^T\gamma^kr(s_k,a_k) \bigg| s_0=s, a_0=a\right].
\end{equation*}
In value-based deep RL implementations, a neural network parameterization $Q_{\psi}(s,a)$ approximates this function, and the policy parameters $\theta$ are updated to maximize $Q_{\psi}(s,a)$ with $a \sim \pi_{\theta}(\cdot|s)$. The widely adopted Soft Actor-Critic (SAC) framework operates on this principle.

Within these algorithms, the policy typically outputs a raw, potentially unbounded action $a \in \SetA$. To interface the RL policy with the underlying optimal control layer, this raw output is passed through a bounded nonlinear mapping $b$ (e.g., an element-wise $\tanh$ function), yielding the abstract action $\bar{a} = b(a) \in \SetAbar$. Consequently, the abstract action set $\SetAbar \subset \Reals{n}$ forms a compact set. For the proposed FAOC framework, we additionally assume that $\SetAbar$ is a solid, convex set (i.e., it possesses a nonempty interior, $\interior \SetAbar \neq \emptyset$), with a standard configuration being the hyperbox $[-1, 1]^{n}$.

\subsection{Parameterized Optimal Control Formulation}

For the mapping algorithm in \Cref{sec:mapping} to successfully couple the RL agent with the OC framework, we require the state-dependent parameter set of the OCP to share the topological properties of the abstract action set $\SetAbar$---namely, it must be compact, solid, convex, and of dimensionality $n$. To establish the mathematical conditions under which the parameter set satisfies these properties, we formulate the parameterized OCP as the following mathematical program:
\begin{align}
    \min_{z \in \SetZ} \quad & f(z, p, x) \label{eq:ocp}                  \\
    \text{s.t.} \quad        & g(z, p, x) \le 0, \nonumber                \\
                             & H_z z + H_p p + H_x x + h = 0\,. \nonumber
\end{align}
Here, $x \in \Reals{n_x}$ is the measured or estimated initial state of the dynamical system, which acts as a fixed parameter establishing the initial conditions. The vector $z \in \Reals{n_z}$ groups the OCP decision variables (e.g., the predicted state and input trajectories over a finite horizon). Finally, $p \in \Reals{n}$ is the action-encoding parameter dictated by the RL agent (e.g., representing a terminal state target).

Depending on the values of the initial state $x$ and the selected parameter $p$, the OCP \eqref{eq:ocp} may or may not possess a feasible solution $z$. To ensure the safe and continuous operation of the system given an initial state~$x$, the RL agent must only select parameters $p$ that guarantee problem feasibility. We formalize the topological properties of the state-dependent feasible parameter set in the following lemma.

\begin{lemma}[Topological Properties of the Parameter Set]
    \label{lem:compactness_of_P}
    For a given initial state $x$, let $\SetP{x}$ denote the effective domain of the parameterized OCP~\eqref{eq:ocp}, defined as the set of parameters $p$ for which a feasible decision vector $z$ exists:
    \begin{align*}
        \SetP{x} \triangleq \{p \in \Reals{n} \mid & \exists z \in \SetZ :
        g(z, p, x) \le 0,                                                  \\ &H_z z + H_p p + H_x x + h = 0\}\,.
    \end{align*}
    Assume that for every dynamically reachable initial state $x$:
    \begin{enumerate}
        \item The set $\SetZ \subseteq \Reals{n_z}$ is nonempty, compact, and convex.
        \item The cost function $f(z, p, x)$ is a real-valued continuous function everywhere on $\SetZ \times \Reals{n}$. The constraint function $g(z, p, x)$ is jointly convex in $(z, p)$ and continuous.
        \item The joint feasible set $\mathcal{S}(x) = \{(z, p) \in \SetZ \times \Reals{n} \mid g(z, p, x) \le 0, \, H_z z + H_p p + H_x x + h = 0\}$ has no directions of recession in $p$.
        \item There exist feasible points $\tilde{z} \in \relint\SetZ$ and $\tilde{p} \in \Reals{n}$ such that $H_z \tilde{z} + H_p \tilde{p} + H_x x + h = 0$. Furthermore, for these points, $g_i(\tilde{z}, \tilde{p}, x) < 0$ for all components $i$. Finally, $\range(H_p) \subseteq H_z \mathbb{T}$, where $\mathbb{T} = \aff\SetZ - \tilde{z}$.
    \end{enumerate}
    Under these conditions, the state-dependent parameter set $\SetP{x}$ is a compact, solid, and convex subset of $\Reals{n}$.
\end{lemma}
\begin{proof}
    By Assumption~1, for any fixed parameter $p$, the feasible subset $\SetZ$ is closed and bounded, ensuring that---since it is nonempty---a minimum of the continuous cost function~$f$ is attained (Weierstrass theorem). Thus, problem solvability directly coincides with feasibility.

    The joint feasible set $\mathcal{S}(x)$ is convex because it is the intersection of the convex set $\SetZ \times \Reals{n}$ with the convex sublevel sets of $g$ and the hyperplanes defined by the affine equality constraints. Because $g$ and the affine relations are continuous, $\mathcal{S}(x)$ is closed. The parameter set $\SetP{x}$ is the linear projection of $\mathcal{S}(x)$ onto the $p$-subspace. Because the linear projection of a convex set is convex, $\SetP{x}$ is a convex set.

    According to \cite[Theorem~9.1]{rockafellar1970convex}, the linear transformation of a closed convex set is closed if the set contains no directions of recession mapped to zero by the transformation. Because $\SetZ$ is compact, $\mathcal{S}(x)$ contains no such directions of recession, guaranteeing that $\SetP{x}$ is closed. By Assumption~3, $\mathcal{S}(x)$ contains no directions of recession in the $p$-subspace. Therefore, its projection $\SetP{x}$ is bounded \cite[Theorem~8.4]{rockafellar1970convex} and, being closed, $\SetP{x}$ is compact.

    Finally, we prove that $\SetP{x}$ is solid, i.e., $\interior \SetP{x} \neq \emptyset$. Consider the perturbation $p = \tilde{p} + \delta p$. By Assumption~4, for any $\delta p$, there exists a corresponding move $\delta z \in \mathbb{T}$ such that $H_z \delta z = -H_p \delta p$. This ensures the equality constraints remain satisfied, as $H_z (\tilde{z} + \delta z) + H_p (\tilde{p} + \delta p) + H_x x + h = 0$. Because $\tilde{z} \in \relint \SetZ$, $\tilde{z} + \delta z$ remains in $\SetZ$ for sufficiently small $\|\delta p\|$. Furthermore, because the components of $g$ are continuous and strictly satisfied at $(\tilde{z}, \tilde{p})$, $g_i(\tilde{z} + \delta z, \tilde{p} + \delta p, x) < 0$ remains true for sufficiently small $\|\delta p\|$. Thus, there exists a full-dimensional neighborhood around $\tilde{p}$ completely contained within $\SetP{x}$. Consequently, $\SetP{x}$ has a nonempty interior.
\end{proof}

The theoretical guarantees of \Cref{lem:compactness_of_P} apply to a broad range of OCP formulations, including those with a nonconvex cost function~$f$. To illustrate a specific instance, we introduce next a parameterized OCP with linear, time-invariant dynamics and a parameter-dependent terminal constraint.

\begin{corollary}[Linear MPC with Terminal Constraint]
    \label{cor:linear_mpc}
    Consider the following parameterized OCP over a prediction horizon of length~$N$:%
    \begin{align}\label{eq:ocp_terminal}
        \min_{x_{0:N}, u_{0:N-1}} \quad & V(x_N) + \sum_{k=0}^{N-1} l(x_k, u_k, p)                                  \\
        \text{s.t.} \quad               & x_{k+1} = A x_k + B u_k, \quad \forall k \in \{0, \dots, N-1\}, \nonumber \\
                                        & x_0 = x, \nonumber                                                        \\
                                        & (x_k, u_k) \in \SetZ_k, \quad \forall k \in \{0, \dots, N-1\}, \nonumber  \\
                                        & x_{N} \in \SetX_N \nonumber                                               \\
                                        & C x_N = p\,. \nonumber
    \end{align}
    Assume the stage-wise sets $\SetZ_k$ and the terminal set $\SetX_N$ are compact, solid, and convex, and that the cost functions $V$ and $l$ are continuous and real-valued everywhere. For every initial state~$x$ for which there exists a strictly feasible trajectory residing in the interior of the sets $\SetZ_k$ and $\SetX_N$, the feasible parameter set $\SetP{x}$ is a compact, solid, and convex subset of $\Reals{n}$ if and only if the matrix $C \begin{bmatrix} B & AB & A^2 B & \dots & A^{N-1} B \end{bmatrix}$ has full row rank.
\end{corollary}

\begin{proof}
    The OCP formulation~\eqref{eq:ocp_terminal} is a direct specialization of \eqref{eq:ocp} with $z = (x_0, u_0, x_1, u_1, \dots, x_{N-1}, u_{N-1}, x_N)$. The set $\SetZ = \prod_{k} \SetZ_k \times \SetX_N$ is compact and convex (Assumption~1). Because $\SetZ$ is also solid, it is nonempty and satisfies $\relint \SetZ = \interior \SetZ$, yielding $\aff \SetZ = \Reals{n_z}$. Furthermore, the cost function satisfies Assumption~2, and the dynamics alongside the terminal constraint constitute the affine equality constraints. Because $\SetX_N$ is bounded, the parameter $p = C x_N$ must also be bounded, preventing any directions of recession in $p$ (Assumption~3).

    Finally, to verify the range condition in Assumption~4, we note that because $\aff \SetZ = \Reals{n_z}$, the subspace of allowable perturbations is $\mathbb{T} = \Reals{n_z}$. Thus, the condition simplifies to $\range(H_p) \subseteq \range(H_z)$, which says that for a fixed initial state $x$, every parameter perturbation $\delta p \in \Reals{n}$ requires the existence of a sequence of input perturbations $\delta U = ( \delta u_{N-1}, \ldots, \delta u_0)$ such that the resulting terminal state perturbation $\delta x_N$ satisfies $C \delta x_N = \delta p$. Unrolling the linear dynamics yields $\delta x_N = R_N \delta U$, where $R_N = \begin{bmatrix} B & AB & \dots & A^{N-1} B \end{bmatrix}$. The range condition thus simplifies to $\delta p = C R_N \delta U$. This condition holds for any $\delta p \in \Reals{n}$ if and only if $\operatorname{rank}(C R_N) = n$.
\end{proof}

\section{Invertible Mapping Between Convex Sets}
\label{sec:mapping}
\noindent
As established in the previous sections, the core requirement of the FAOC framework is to map an abstract action $\bar{a} \in \SetAbar$ into the state-dependent feasible parameter set $\SetP{x}$. Both of these sets are compact, solid, and convex. Furthermore, to avoid ``action aliasing''---i.e., different actions mapping to the same parameter---this mapping must be invertible. An invertible mapping additionally allows for more flexibility in the use of auxiliary RL techniques, e.g., symmetric data augmentation, or including demonstrations in the replay buffer. Finally, because the mapping resides within the real-time closed-loop control cycle, it must be computationally efficient and scale favorably with the dimension~$n$.

In this section, we develop a family of optimization-based mapping algorithms that satisfy these requirements. To present these results as generically as possible, we momentarily abstract away from the specific RL and OCP domains. We define a compact, solid, and convex base set $\SetX \subset \Reals{n}$ (from which we map) and a corresponding target set $\SetY \subset \Reals{n}$ (into which we map). In the context of the FAOC framework, these correspond to $\SetX \triangleq \SetAbar$ and $\SetY \triangleq \SetP{x}$, respectively. Note that we temporarily drop the explicit parametric state dependence of $\SetY$ for brevity, but will reconsider it in \Cref{ssec:implicit}.

\subsection{A Generic Mapping Algorithm}
\label{ssec:generic_mapping}

\begin{algorithm}[t!]
    \small
    \caption{Invertible Mapping Between Convex Sets}\label{alg:basic}
    \begin{algorithmic}[1]
        \Require Base set~$\SetX \subset \Reals{n}$ and target set~$\SetY \subset \Reals{n}$ (both compact, solid, and convex), interior points $x_c \in \interior\SetX$ and $y_c \in \interior\SetY$, point to map $x \in \SetX$, and a continuous, invertible, positively homogeneous directional transformation $\phi: \Reals{n} \to \Reals{n}$.
        \Ensure $y = \map(x) \in \SetY$, where the mapping $\map$ is invertible.
        \If {$x = x_c$}
        \State $y \gets y_c$
        \Else
        \State $d \gets x - x_c$
        \State Compute $\alpha \ge 1$ such that $x_c + \alpha d \in \bd\SetX$
        \State $d' \gets \phi(d)$
        \State Compute $\beta > 0$ such that $y_c + \beta d' \in \bd\SetY$
        \State $y \gets y_c + \frac{\beta}{\alpha}\,d'$
        \EndIf
    \end{algorithmic}
\end{algorithm}

The foundation of our approach is the radial scaling methodology detailed in \Cref{alg:basic}. The algorithm operates by determining the boundary scaling factor of a ray originating from an interior point of the base set, applying a bijective directional transformation $\phi(\cdot)$, and proportionally scaling the transformed ray relative to an interior point of the target set. A simplified variant of this radial mapping was previously utilized in \cite{Yuan2023Action}, where the authors formally established its continuity. However, their formulation was restricted to cases where the base set is a hypercube, the directional transformation is the identity ($\phi(d) = d$), and the interior point of the target set is fixed at the origin.

A critical requirement for integrating this mapping within the FAOC framework is that the transformation must be reversible. This ensures that any feasible parameter chosen by the optimal controller, or provided via expert demonstration, can be uniquely projected back into the abstract action space of the RL agent, avoiding action aliasing. We formally establish the invertibility of the proposed radial mapping in the following proposition.

\begin{proposition}[Invertibility of the Generic Mapping]\label{prop:Finv_Alg1}
    Given a continuous, invertible, and positively homogeneous directional transformation~$\phi$ and two compact, solid, and convex sets $\SetX \subset \Reals{n}$ and $\SetY \subset \Reals{n}$, \Cref{alg:basic} computes a vector-valued map $\map : \SetX \to \SetY$ such that $\map^{-1}(\map(x)) = x$ for each $x \in \SetX$ and $\map(\map^{-1}(y)) = y$ for each $y \in \SetY$, i.e., $\map$ is bijective. Here, $\map^{-1} : \SetY \to \SetX$ is the inverse mapping obtained by exchanging the roles of $(\SetX, x_c, x)$ and $(\SetY, y_c, y)$ in \Cref{alg:basic} and replacing the directional transformation $\phi$ with its inverse~$\phi^{-1}$.
\end{proposition}
\begin{proof}
    For $x = x_c$, the trivial mapping $y = y_c$ is directly inverted. For $x \neq x_c$, because $x_c \in \interior\SetX$ and $\SetX$ is compact and convex, any ray originating from $x_c$ intersects the boundary $\bd\SetX$ at a unique, finite scaling factor $\alpha \ge 1$. Similarly, because $y_c \in \interior\SetY$ and $\phi(d) \neq 0$, the ray defined by $d'$ hits $\bd\SetY$ at a unique, finite $\beta > 0$. The mapped point $y$ is bounded by the segment $[y_c, y_c + \beta\, d']$ because $\alpha \ge 1$, ensuring $y \in \SetY$.

    To apply the inverse mapping $\map^{-1}$ to $y$, we extract the direction $v' = y - y_c = \frac{\beta}{\alpha}\, d'$. The scaling factor $\tilde{\alpha}$ for $v'$ required to reach $\bd\SetY$ is determined as $\tilde{\alpha} = \alpha \ge 1$, since $y_c + \tilde{\alpha}\, v' = y_c + \beta\, d' \in \bd\SetY$. The inverse directional transformation yields $v = \phi^{-1}(v') = \phi^{-1}\big(\frac{\beta}{\alpha}\, \phi(d)\big)$, which results in $v = \frac{\beta}{\alpha}\, d$ since the inverse $\phi^{-1}$ is positively homogeneous. Scaling to $\bd\SetX$ requires $\tilde{\beta} = \frac{\alpha^2}{\beta}$, as $x_c + \tilde{\beta}\, v = x_c + \alpha\, d \in \bd\SetX$. Reconstructing the original point yields $x = x_c + \frac{\tilde{\beta}}{\tilde{\alpha}}\, v = x_c + \frac{\alpha}{\beta} \frac{\beta}{\alpha}\, d = x_c + d$, recovering the original~$x$.
\end{proof}

Implementing \Cref{alg:basic} requires computing the interior points ($x_c, y_c$) and evaluating the boundary scaling factors ($\alpha, \beta$). An interior point can be obtained via standard convex optimization techniques (e.g., the analytic center or the Chebyshev center, which maximizes the radius of an inscribed ball) \cite{boyd2004convex}. The boundary scaling factors are determined by solving simple convex line-search problems along the given directions:
\begin{align*}
    \alpha = \max \{\alpha \mid x_c + \alpha\, d \in \SetX \}, \enspace \beta = \max \{\beta \mid y_c + \beta\, d' \in \SetY \}.
\end{align*}
In the prevalent case where $\SetX$ and $\SetY$ are polytopes defined by H-representations ($\SetX = \{x \mid A_x x \le b_x\}$, $\SetY = \{y \mid A_y y \le b_y\}$), the line-search problems collapse to computationally trivial ratio tests over the bounding half-spaces. For instance, the scaling factor $\alpha$ is analytically given by
\begin{equation*}
    \alpha = \min_{i \in \{j \mid a_{x, j}^\top d > 0\}} \frac{b_{x, i} - a_{x, i}^\top x_c}{a_{x, i}^\top d}\,,
\end{equation*}
where $a_{x, i}^\top$ denotes the $i$-th row of matrix $A_x$, and $b_{x, i}$ the corresponding element of $b_x$.

If the directional transformation is simply chosen as the identity function ($\phi(d) = d$), e.g., as in \cite{Yuan2023Action}, the algorithm remains topologically valid but becomes geometrically ``blind.'' Specifically, if the target set $\SetY$ possesses a significantly different aspect ratio or orientation than the base set $\SetX$, uniformly distributed points in $\SetX$ will accumulate unnaturally in the narrower regions of $\SetY$. As illustrated in the second plot of \Cref{fig:mapping_comparison}, this distortion requires an excessive number of samples during learning to adequately explore the action space. The next section proposes a specific, area-aware transformation~$\phi$ that mitigates this unwanted effect.

\begin{remark}[Scope of Feasibility]
    It is important to distinguish the scope of the mapping $\map$ from overall closed-loop guarantees. \Cref{alg:basic} guarantees that for the current state $x$, the mapped parameter yields a feasible OCP \eqref{eq:ocp}. Recursive feasibility (i.e., guaranteed existence of $p^+ \in \SetP{x^+}$ at subsequent time steps) is not an intrinsic property of the mapping itself, but is inherited when the underlying OCP incorporates standard terminal ingredients, such as a control-invariant terminal set.
\end{remark}

\begin{figure*}[!t]
    \centering
    \includegraphics[trim=0.0cm 0.0cm 0.0cm 0.0cm,clip,width=\textwidth]{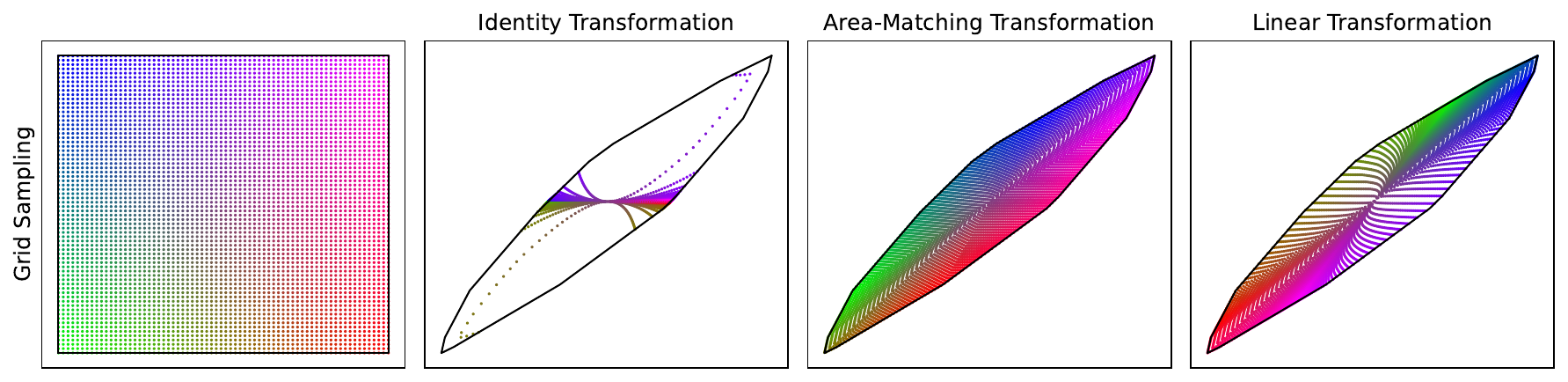}
    \caption{Visualization of the proposed mapping approaches.
        \emph{First plot:} The base set $\SetX$ (a unit square in $\Reals{2}$) is sampled using a structured grid pattern.
        \emph{Second plot:} Naive mapping utilizing an identity directional transformation ($\phi(d_x)=d_x$). This results in significant point accumulation, as the highly elongated geometry of the target set $\SetY$ differs severely from $\SetX$ (note that for visualization purposes, the x:y axes scaling ratio of the $\SetY$ plots is 1:60).
        \emph{Third plot:} Mapping utilizing the area-matching transformation $\phi$. While strictly valid only in $\Reals{2}$, this method matches the marginal angular distributions of the sets' areas under uniform sampling, significantly mitigating angular point accumulation.
        \emph{Fourth plot:} Mapping utilizing an invertible linear transformation ($\phi(d_x) = \Phi d_x$). By employing affine surrogate sets as a motivated geometric heuristic, this approach effectively warps search directions to approximate the target set's aspect ratio, mitigating point accumulation in arbitrary dimensions.}
    \label{fig:mapping_comparison}
\end{figure*}

\subsection{Area-Matching Transformation in 2D}\label{ssec:area_matching}

To mitigate point accumulation, the directional transformation~$\phi$ must account for the geometric disparities between the base and target sets. In this section, we derive a transformation~$\phi$ for $n=2$, which matches the marginal angular distributions of the sets' areas, thereby mitigating point accumulation under the mapping~$\map$. To formalize this geometric relationship, we construct the following area-based measure.
\begin{definition}[Normalized Area Function]
    Let $\SetX \subset \Reals{2}$ be a compact, solid, and convex set, $A_{\SetX} = \int_{\SetX}\, dx$ be its total area, and $x_c$ be an interior point. For any angle $\xi \in [0, 2\pi]$, we define the angular wedge:
    \begin{align*}
        \SetX(\xi) \triangleq \SetX \cap \bigg\{ x_c + \mu \begin{bmatrix} \cos(\theta) \\ \sin(\theta) \end{bmatrix} \mid \mu \ge 0, \, \theta \in [0, \xi] \bigg\}\,.
    \end{align*}
    The normalized area function is then $a_{\SetX}(\xi) \triangleq A_{\SetX(\xi)} / A_{\SetX}$.
\end{definition}

This function acts as a marginal Cumulative Distribution Function (CDF) representing a 2D uniform distribution over $\SetX$ with respect to the polar angle originating from $x_c$. Specifically, $a_{\SetX}(0) = 0$, $a_{\SetX}(2\pi) = 1$, and if a random vector $X$ is distributed uniformly over $\SetX$, then $\Prob(X \in \SetX(\xi)) = a_{\SetX}(\xi)$. Furthermore, because $\SetX$ is solid and convex, $a_{\SetX}$ is continuous and strictly increasing.

To construct a transformation that accounts for geometric disparities, we define $\phi: \Reals{2} \to \Reals{2}$ using the normalized area functions of the respective sets. First, we recall a standard result regarding the invertibility of such functions.

\begin{lemma}[Monotone Invertibility {\cite[Chap.~12, Thm.~3]{spivak2008calculus}}]
    \label{lem:strincr}
    Let $f: [a, b] \to \Reals{}$ be a strictly increasing and continuous function. Then its inverse function $f^{-1}$ exists, is strictly increasing, and is continuous on the domain $[f(a), f(b)]$.
\end{lemma}

Because the normalized area function $a_{\SetY}(\xi)$ is continuous and strictly increasing on $[0, 2\pi]$, its inverse $a_{\SetY}^{-1}$ is well-defined on $[0,1]$ due to \Cref{lem:strincr}. This allows us to define the specific directional transformation for the 2D area-matching approach.

\begin{definition}[Area-Matching Directional Transformation]
    \label{def:uniform_phi}
    For any base direction $d \in \Reals{2}$, let $\xi = \arg(d)$ be its polar angle. The area-matching directional transformation $\phi: \Reals{2} \to \Reals{2}$ is defined as:
    \begin{equation}\label{eq:uniform_phi}
        \phi(d) \triangleq \|d\| \begin{bmatrix} \cos(\upsilon) \\ \sin(\upsilon) \end{bmatrix}, \quad \text{where} \quad \upsilon = a_{\SetY}^{-1}(a_{\SetX}(\xi))\,.
    \end{equation}
\end{definition}

In order to guarantee that integrating \Cref{def:uniform_phi} into \Cref{alg:basic} yields a valid and invertible mapping between $\SetX$ and $\SetY$, the transformation $\phi$ must satisfy the topological conditions established in \Cref{prop:Finv_Alg1}.

\begin{proposition}[Area-Matching Transformation Properties]
    \label{prop:phi_properties}
    The area-matching directional transformation~$\phi$ defined in \eqref{eq:uniform_phi} is continuous, invertible, and positively homogeneous.
\end{proposition}
\begin{proof}
    By definition, the normalized area functions $a_{\SetX}$ and $a_{\SetY}$ are continuous and strictly increasing in the interval $[0, 2\pi]$. By \Cref{lem:strincr}, $a_{\SetY}^{-1}$ is also continuous and strictly increasing. Since the composition of strictly increasing continuous functions is strictly increasing and continuous, the angular mapping $\upsilon(\xi)$ is a continuous bijection on $[0, 2\pi)$. Combined with the norm preservation ($\|d'\| = \|d\|$), the overall transformation $\phi: \Reals{2} \to \Reals{2}$ is continuous and bijective. Furthermore, let $v = c d$, $c > 0$. Scaling a vector by a positive scalar preserves its polar angle, i.e., $\xi' = \arg(c d) = \xi$. Because $\xi$ is invariant under positive scaling, the mapped target angle $\upsilon$ remains unchanged. Applying the definition of $\phi$ to the scaled vector yields:
    \begin{equation*}
        \phi(c d) = \|c d\| \begin{bmatrix} \cos(\upsilon) \\ \sin(\upsilon) \end{bmatrix} = c \, \|d\| \begin{bmatrix} \cos(\upsilon) \\ \sin(\upsilon) \end{bmatrix} = c \phi(d)\,.
    \end{equation*}
    Thus, $\phi$ is positively homogeneous.
\end{proof}

Having established that integrating this area-matching transformation into the generic mapping algorithm preserves strict invertibility (due to the established properties of $\phi$ and \Cref{prop:Finv_Alg1}), we now formalize its primary objective: mitigating point accumulation. We demonstrate that if the abstract action space is uniformly sampled, the resulting mapped parameters are distributed such that they respect the area proportions of the target set.

\begin{proposition}[Preservation of Marginal Area Proportions]\label{prop:marginal}
    Assume that a random vector $X$ is distributed uniformly over the compact, solid, and convex base set $\SetX \subset \Reals{2}$. Thus, its marginal angular CDF matches the normalized area function:
    \begin{equation*}
        P_{\SetX}(\xi) \triangleq \Prob(X \in \SetX(\xi)) = a_{\SetX}(\xi) \text{ for } \xi \in [0, 2\pi]\,.
    \end{equation*}
    Applying the mapping algorithm with the directional transformation $\phi$ defined by \eqref{eq:uniform_phi} yields a mapped random vector $Y = \map(X) \in \SetY$ whose marginal angular CDF identically matches the normalized area function of the target set, i.e., $P_{\SetY}(\upsilon) = a_{\SetY}(\upsilon)$.
\end{proposition}

\begin{proof}
    By the geometric nature of the radial scaling algorithm, probability mass along any angular wedge originating from $x_c \in \SetX$ is bijectively mapped to the corresponding angular wedge originating from $y_c \in \SetY$. Thus, the mapped marginal CDF satisfies $P_{\SetY}(\upsilon) = P_{\SetX}(\xi(\upsilon))$.

    Substituting the inverse angular relation $\xi(\upsilon) = a_{\SetX}^{-1}(a_{\SetY}(\upsilon))$ defined in \eqref{eq:uniform_phi}, and applying the uniform premise $P_{\SetX}(\xi) \equiv a_{\SetX}(\xi)$, yields:
    \begin{equation*}
        P_{\SetY}(\upsilon) = a_{\SetX}\Big( a_{\SetX}^{-1}\big(a_{\SetY}(\upsilon)\big) \Big) = a_{\SetY}(\upsilon)\,.
    \end{equation*}
    Because $a_{\SetY}(\upsilon)$ is exactly the marginal angular CDF of a uniform distribution over $\SetY$, the mapping successfully preserves the marginal area proportions of the target set.
\end{proof}

\begin{remark}[Scope of the Area-Matching Transformation]
    \label{rem:area_matching_scope}
    Two aspects of \Cref{prop:marginal} warrant clarification:
    \begin{enumerate}
        \item \textbf{Marginal vs.~Joint Distribution:} While \Cref{prop:marginal} guarantees that the marginal angular distribution of $Y = \map(X)$ matches the normalized area function of $\SetY$, it does not imply a strictly uniform joint distribution over $\SetY$, as radial densities may still undergo local distortion depending on boundary curvature.
        \item \textbf{RL Policy Sampling:} In practice, continuous RL algorithms (such as SAC) do not generate uniform distributions over the abstract set $\SetAbar$, but rather squashed Gaussian distributions. \Cref{prop:marginal} evaluates the uniform case to isolate the pure geometric effect of the mapping baseline. The transformation acts as a geometric pre-conditioner, preventing structural point accumulation independent of the policy's evolving density.
    \end{enumerate}
\end{remark}

From a computational perspective, explicitly evaluating the normalized area function and its inverse may appear expensive due to the trigonometric functions and integrations involved. However, when $\SetX$ and $\SetY$ are convex polygons (e.g., defined by counterclockwise-ordered vertex sets), $a_{\SetX}(\xi)$ and $a_{\SetY}^{-1}(a_{\SetX}(\xi))$ can be computed analytically using sequential signed triangular area calculations. This geometric approach completely circumvents explicit angle computation via trigonometric evaluations, rendering the 2D area-matching transformation highly efficient for real-time applications. As illustrated in \Cref{fig:mapping_comparison} in the third plot, utilizing this area-matching transformation effectively eliminates point accumulation.

While the area-matching transformation mitigates point accumulation in 2D, its reliance on exact volumetric integration renders it computationally intractable for arbitrary dimensions $n > 2$. We now present an approximate, scalable alternative based on a linear transformation.

\subsection{Linear Transformation in Arbitrary Dimensions}\label{ssec:linear_mapping}

In order to maintain real-time applicability while accounting for high-dimensional geometric distortion, we propose a scalable linear directional transformation, $\phi(d) = \Phi d$, as a motivated geometric heuristic.

Rather than attempting to preserve the exact joint density, the matrix $\Phi \in \Reals{n \times n}$ is constructed to warp search directions according to the principal axes and aspect ratios of the sets. Because any invertible linear map is continuous, bijective, and positively homogeneous, this choice satisfies the topological requirements of \Cref{prop:Finv_Alg1} while significantly reducing point accumulation.



To formally construct the transformation matrix $\Phi$, we rely on the concept of \emph{affinely related surrogate sets}.

\begin{assumption}[Affinely Related Surrogate Sets]\label{assum:surrogate}
    Let $\widetilde{\SetX}$ be a solid and convex surrogate set that approximates the geometric proportions of the base set $\SetX$. Suppose we are given an invertible affine mapping defined by a matrix $\Phi \in \Reals{n \times n}$ alongside an offset vector $c_\Phi \in \Reals{n}$. The image of the base surrogate under this mapping produces a solid and convex set, $\widetilde{\SetY} = \Phi\widetilde{\SetX} + c_\Phi$, acting as a surrogate that captures the geometric proportions of the target set $\SetY$.
\end{assumption}

If \Cref{assum:surrogate} holds, substituting $\phi(d) = \Phi d$ into \Cref{alg:basic} ensures that the radial scaling algorithm distributes points in a manner that closely mimics the target set's geometry, effectively mitigating point accumulation. The challenge remains to systematically compute $\Phi$. A highly practical choice for the surrogate sets $\widetilde{\SetX}$ and $\widetilde{\SetY}$ is to define them as solid ellipsoids. This geometric structure admits a direct analytical computation of the linear transformation matrix $\Phi$.

\begin{proposition}[Ellipsoidal Surrogates]\label{prop:ellipsoid_surrogate}
    Let the surrogate sets $\widetilde{\SetX}$ and $\widetilde{\SetY}$ be solid ellipsoids, characterized by centers $q_{\widetilde{\SetX}}$, $q_{\widetilde{\SetY}}$ and invertible shape matrices $Q_{\widetilde{\SetX}}$, $Q_{\widetilde{\SetY}}$, respectively. The linear transformation matrix satisfying \Cref{assum:surrogate} is given by $\Phi = Q_{\widetilde{\SetY}} Q_{\widetilde{\SetX}}^{-1}$.
\end{proposition}

\begin{proof}
    The solid ellipsoidal surrogates are explicitly defined as:
    \begin{align*}
        \widetilde{\SetX} & = \{ x \in \Reals{n} \mid x = Q_{\widetilde{\SetX}} u + q_{\widetilde{\SetX}}, \enspace \|u\|_2 \le 1 \} \,, \\
        \widetilde{\SetY} & = \{ y \in \Reals{n} \mid y = Q_{\widetilde{\SetY}} v + q_{\widetilde{\SetY}}, \enspace \|v\|_2 \le 1 \} \,.
    \end{align*}
    To construct the affine map connecting the surrogate sets, we isolate the unit vector. Because the ellipsoids are solid, their shape matrices are invertible. For any $x \in \widetilde{\SetX}$, the corresponding unit vector is uniquely given by $u = Q_{\widetilde{\SetX}}^{-1}(x - q_{\widetilde{\SetX}})$. Substituting $u$ into the definition of $\widetilde{\SetY}$ as a feasible vector $v$ yields the corresponding mapped point $y \in \widetilde{\SetY}$:
    \begin{equation*}
        y = Q_{\widetilde{\SetY}} Q_{\widetilde{\SetX}}^{-1}(x - q_{\widetilde{\SetX}}) + q_{\widetilde{\SetY}}\,.
    \end{equation*}
    This establishes the affine point-wise mapping $y = \Phi x + c_\Phi$, where the linear transformation matrix is $\Phi = Q_{\widetilde{\SetY}} Q_{\widetilde{\SetX}}^{-1}$ and the offset is $c_\Phi = q_{\widetilde{\SetY}} - \Phi q_{\widetilde{\SetX}}$. Because $\Phi$ is the product of two invertible matrices, it is also invertible, thus fulfilling all conditions of \Cref{assum:surrogate}.
\end{proof}

In the prevalent case where the true sets $\SetX$ and $\SetY$ are explicitly known polytopes defined by H-representations, an effective way to obtain these ellipsoidal surrogates is by computing their Maximum Volume Inscribed Ellipsoids (MVIEs). This can be achieved efficiently by solving standard convex Semidefinite Programs (SDPs)~\cite{boyd2004convex}.

Utilizing this transformation matrix $\Phi$ to warp the search directions approximately aligns the mapped distribution with the geometric proportions of the target set. As illustrated in the fourth plot of \Cref{fig:mapping_comparison}, this approach significantly reduces point accumulation and approximately preserves the proportional density of the base distribution, all without the computational burden of multidimensional integration. However, constructing $\Phi$ requires explicitly defining the surrogate set $\widetilde{\SetY}$, which necessitates full geometric knowledge of the target set $\SetY$. In hierarchical control schemes like FAOC, the state-dependent feasible parameter set is frequently defined only \emph{implicitly} via the OCP constraints. The following section extends this linear transformation methodology to implicitly defined, state-dependent target sets, circumventing the need for computationally prohibitive explicit geometric representations and thereby enabling real-time applicability.

\subsection{Mapping for Implicitly Defined Target Sets}\label{ssec:implicit}

To efficiently execute the linear transformation mapping without relying on an explicit geometric representation of the target set, we consider the case where it is defined as the affine image of a compact and convex set $\SetW \subset \Reals{n_w}$ with dimension $n_w > n$, i.e.,
\begin{equation}\label{eq:implicit_Y}
    \SetY = E \SetW + e\,,
\end{equation}
where $E \in \Reals{n \times n_w}$ is a projection matrix and $e \in \Reals{n}$ is an offset vector. By \eqref{eq:implicit_Y}, the set $\SetY$ is inherently compact and convex. However, for the mapping in \Cref{alg:basic} to apply, $\SetY$ must also be solid. To guarantee this topological property, we make the following assumption.

\begin{assumption}[Affine Hull Projection]\label{ass:implicit}
    The image of the affine hull of $\SetW$ under $E$ spans the entire target space, i.e., $E(\aff \SetW) = \Reals{n}$.
\end{assumption}

A necessary consequence of this assumption is that the matrix~$E$ must possess full row rank. Under this rank condition, a sufficient (though not necessary) condition for \Cref{ass:implicit} to hold is that $\SetW$ is solid.

To relate this generic geometric setup to the OCP~\eqref{eq:ocp} previously discussed, we provide the following example.

\begin{example}[Implicit Target Set for Generic OCP]
    \label{ex:ocp_impl}
    Consider the generic parameterized OCP defined in~\eqref{eq:ocp}. To map this formulation into the structure of \eqref{eq:implicit_Y}, we must isolate the parameter $p$. To this end, we assume that the subset of parameter-dependent affine equality constraints can be rearranged into the form $B_p p = B_z z + B_x x + b$, where the square matrix $B_p \in \Reals{n \times n}$ is invertible. This allows us to isolate the parameter as $p = B_p^{-1}(B_z z + B_x x + b)$ and thus to express the inequality constraints $g(z,p,x) \leq 0$ solely in $(z, x)$, abbreviated as $\tilde{g}(z, x) \leq 0$ below. Note that in prevalent formulations for linear MPC, e.g., \eqref{eq:ocp_terminal}, where $p = C x_N$, this natural separation holds exactly, with $B_p$ simply being the identity matrix.
    
    Let the parameter-independent equality constraints be denoted by $A_z z + A_x x + a = 0$. To match the proposed implicit structure, we now define $\SetW$ as the state-dependent set
    \begin{equation*}
        \SetW(x) \triangleq \{ z \in \SetZ \mid \tilde{g}(z,x) \le 0, \, A_z z + A_x x + a = 0 \}\,.
    \end{equation*}
    Note that since this set contains explicit equality constraints, it generally possesses an empty interior in $\Reals{n_z}$.

    By definition, a parameter $p$ belongs to the state-dependent feasible parameter set $\SetP{x}$ of the OCP if and only if there exists a $z \in \SetW(x)$ such that $p = B_p^{-1}(B_z z + B_x x + b)$. Thus, the parameter set takes the form $\SetP{x} = E \SetW(x) + e(x)$, where the state-independent projection matrix is $E = B_p^{-1} B_z$, and the state-dependent offset is $e(x) = B_p^{-1}(B_x x + b)$. Provided the original OCP satisfies the assumptions of \Cref{lem:compactness_of_P} at the initial state~$x$, the resulting parameter set $\SetP{x}$ is guaranteed to be compact, solid, and convex.
\end{example}

Having established the implicit definition of the target set~$\SetY$ in \eqref{eq:implicit_Y}, we now demonstrate how to execute \Cref{alg:basic} by computing its three required components---an interior point $y_c \in \interior\SetY$, the linear transformation matrix~$\Phi$, and the boundary scaling factor~$\beta$---directly from $\SetW$ and the affine mapping $(E, e)$, completely circumventing the need to extract a computationally demanding explicit representation of $\SetY$.

\subsubsection{Computing an Interior Point}
The first requirement of \Cref{alg:basic} is to locate an interior point $y_c$ within the implicitly defined target set $\SetY$. The following proposition establishes that this can be achieved purely by projecting a relative interior point of the higher-dimensional base set $\SetW$, circumventing any geometric operations in the target space.

\begin{proposition}[Interior Point of an Implicitly Defined Set]\label{prop:implicit_interior}
    Given a convex set $\SetW \subset \Reals{n_w}$ and a point $w_c \in \relint\SetW$, the mapped point $y_c = E w_c + e$ is an interior point of $\SetY = E \SetW + e$, provided \Cref{ass:implicit} holds.
\end{proposition}
\begin{proof}
    According to \cite[Thm.~6.6]{rockafellar1970convex}, the relative interior of a convex set is preserved under linear transformations. Thus, the relative interior of the target set is given by $\relint\SetY = E(\relint\SetW) + e$. Since $w_c \in \relint\SetW$, it follows that $y_c \in \relint\SetY$. By \Cref{ass:implicit}, the affine hull of the target set is the entire space, thus $\relint\SetY = \interior\SetY$ and $y_c \in \interior\SetY$.
\end{proof}

A robust method to compute a point $w_c \in \relint\SetW$ is to explicitly parameterize the affine hull of the set. Any point $w \in \aff\SetW$ can be uniquely expressed using a reduced coordinate vector $v$ as $w = F v + w_p$, where the columns of the matrix $F$ form a basis for the linear subspace parallel to $\aff\SetW$, and $w_p \in \aff\SetW$ is a particular translation vector.

By applying this coordinate transformation, we represent the problem in a lower-dimensional space, defining the reduced set:
\begin{equation*}
    \SetV \triangleq \{ v \mid F v + w_p \in \SetW \}\,.
\end{equation*}
Because this parameterization explicitly eliminates the affine equality constraints, the resulting set $\SetV$ is solid ($\interior\SetV \neq \emptyset$). Consequently, one can compute an interior point $v_c \in \interior\SetV$ using standard convex optimization techniques, such as finding the Chebyshev center. The required relative interior point is finally recovered via $w_c = F v_c + w_p \in \relint\SetW$.

To illustrate this coordinate reduction in a practical OC setting, we provide the following example.

\begin{example}[Interior Point Computation for Linear MPC]\label{ex:ocp_interior}
    Consider the parameterized OCP~\eqref{eq:ocp_terminal}, where the generic decision vector is $z = (x_{0:N}, u_{0:N-1})$, composed of the state sequence $x_{0:N}$ and the input sequence $u_{0:N-1}$ (note that we have ordered the components of $z$ differently than in the proof of \Cref{cor:linear_mpc}). Because the stage-wise sets $\SetZ_k$ and the terminal set $\SetX_N$ are assumed to be solid, the affine equality constraints defining the affine hull of $\SetW(x)$ stem exclusively from the initial state condition and the system dynamics.

    We can compactly express these equalities as a linear system:
    \begin{equation*}
        \bar{A} x_{0:N} + \bar{B} u_{0:N-1} = \bar{C} x\,.
    \end{equation*}
    Because the matrix $\bar{A}$ represents the causal state propagation, it is inherently lower triangular with identity matrices on the diagonal, and thus always invertible. This allows us to explicitly express the state trajectory in terms of the initial state $x$ and the control sequence $u_{0:N-1}$:
    \begin{equation*}
        x_{0:N} = \bar{A}^{-1}(\bar{C} x - \bar{B} u_{0:N-1})\,.
    \end{equation*}

    This explicit relation serves as the exact analogue to the generic affine parameterization ($w = F v + w_p$), where the control sequence $u_{0:N-1}$ acts as the reduced coordinate vector~$v$. Substituting this expression for $x_{0:N}$ into the constraints represents the problem in the reduced input space, defining
    \begin{equation*}
        \SetV(x) \triangleq \left\{ u_{0:N-1} \;\middle|\;
        \begin{array}{l}
            (x_k, u_k) \in \SetZ_k, \; x_N \in \SetX_N \\
            \text{with } x_{0:N} = \bar{A}^{-1}(\bar{C} x - \bar{B} u_{0:N-1})
        \end{array}
        \right\}\,.
    \end{equation*}
    Provided the initial state $x$ admits a strictly feasible trajectory (as assumed in \Cref{cor:linear_mpc}), there exists at least one control sequence that strictly satisfies the constraints imposed by the solid bounding sets $\SetZ_k$ and $\SetX_N$. Consequently, the resulting set $\SetV(x)$ is guaranteed to be a solid set ($\interior\SetV(x) \neq \emptyset$). One can therefore compute an interior point $u_{c, 0:N-1} \in \interior\SetV(x)$. Finally, the corresponding state trajectory $x_{c, 0:N}$ is computed via dynamic propagation, yielding the required relative interior point $w_c = (x_{c, 0:N}, u_{c, 0:N-1}) \in \relint\SetW(x)$.
\end{example}

Numerical experience indicates that relying on interior points computed via \Cref{prop:implicit_interior} can occasionally introduce numerical sensitivities within the mapping procedure detailed in \Cref{alg:basic}. For instance, if the target set $\SetY$ exhibits a highly skewed aspect ratio or narrow regions, an arbitrarily positioned interior point might fall disproportionately close to a boundary, thereby compromising the reliable computation of the boundary scaling factor~$\beta$. To mitigate this, we propose two optimization-based formulations designed to compute an interior point~$y_c^*$ ``deep'' within $\SetY$ by explicitly maximizing an inscribed norm ball, all without requiring an explicit representation of $\SetY$.

The first formulation utilizes a vertex-based representation.

\begin{proposition}[Vertex-Based Inscribed Ball]\label{prop:vertex_inscribed}
    Let $\Ball(r) \subset \Reals{n}$ be the closed ball of radius $r > 0$ under a polyhedral norm~$\|\cdot\|_P$, i.e., $\Ball(r) = \{ q \mid \| q \|_P \le r \}$, and let $\{ q_j \}_{j=1}^{m}$ denote the finite set of vertices of the unit ball $\Ball(1)$. Assume we are given a nonempty convex set~$\SetW$ and an affine mapping defined by $(E, e)$ such that $\SetY = E \SetW + e$. Let $G \in \Reals{n \times n}$ be a user-defined invertible shape matrix, chosen to align the inscribed ball with the geometry of the target set. The radius $r^*$ of the largest scaled ball $G \Ball(r^*) + y_c$ inscribed in $\SetY$ is obtained by solving the following convex optimization problem:%
    \begin{align}\label{eq:poly_norm_ball}
        r^* = \max_{r, y_c, w_{1:m}} \quad & r                                                                         \\
        \text{s.t.} \quad                  & E w_j + e = r G q_j + y_c, \quad \forall j \in \{1, \dots, m\}, \nonumber \\
                                           & w_j \in \SetW, \quad \forall j \in \{1, \dots, m\}\,. \nonumber
    \end{align}
    If $r^* > 0$, then $y_c^* \in \interior\SetY$.
\end{proposition}

\begin{proof}
    Because $G \Ball(r) + y_c$ is a convex polytope, the inclusion condition $G \Ball(r) + y_c \subseteq \SetY$ holds if and only if all vertices of the former are contained within the latter. Thus, it is necessary and sufficient that $r G q_j + y_c \in \SetY$ for all $j \in \{1, \dots, m\}$. By the definition of $\SetY$, this implies that every such mapped vertex must correspond to at least one valid pre-image $w_j \in \SetW$ such that $E w_j + e = r G q_j + y_c$. Maximizing $r$ yields the largest such inscribed ball. Since $\SetW$ is nonempty, $r^* \ge 0$. If $r^* > 0$, the invertibility of $G$ guarantees that the inscribed ball possesses a nonempty interior, proving that the optimal center point satisfies $y_c^* \in \interior \SetY$.
\end{proof}

For a 1-norm ball, this formulation introduces $m = 2n$ vertices, whereas an $\infty$-norm ball yields $m = 2^n$ vertices. Consequently, the computational burden of \eqref{eq:poly_norm_ball} restricts its real-time applicability to target sets of small dimension~$n$.

To circumvent the scaling of the number of constraints with the target set dimension $n$, we present an alternative formulation that builds upon the concept of lifting. This approach naturally extends to non-polyhedral norm balls and restricts the number of constraints to the order required to describe $\SetW$. While this method provides only a lower bound to the true maximum radius~$r^*$, it remains computationally tractable for higher dimensions~$n$. For brevity, we detail this approach specifically for a polyhedral $\SetW$, noting that the principle generalizes to arbitrary convex sets.

\begin{proposition}[Lifting-Based Inscribed Ball]\label{prop:dual_inscribed}
    Let $\SetB(r) \subset \Reals{n}$ be the closed ball of radius $r > 0$ under an arbitrary norm~$\|\cdot\|$. Assume $\SetW$ is a nonempty polyhedral set given by $\SetW = \{ w \in \Reals{n_w} \mid a_j^\top w \le b_j, \, j = 1, \ldots, n_c\}$ and its affine image is $\SetY = E \SetW + e$. Let $G \in \Reals{n \times n}$ be a user-defined invertible shape matrix. A lower bound $\underline{r} \le r^*$ to the radius of the largest scaled ball $G \SetB(r^*) + y_c$ inscribed in $\SetY$ is computed by the following convex optimization problem over the variables $r \in \Reals{}$, $w_c \in \Reals{n_w}$, and $L \in \Reals{n_w \times n}$:
    \begin{align}\label{eq:dual_norm_ball}
        \underline{r} = \max_{r, L, w_c} \quad & r                                                                                           \\
        \text{s.t.} \quad                      & E L = r G, \nonumber                                                                        \\
                                               & \|L^\top a_j\|_* + a_j^\top w_c \le b_j, \quad \forall j \in \{1, \dots, n_c\}\,, \nonumber
    \end{align}
    where $\|\cdot\|_*$ denotes the dual norm of $\|\cdot\|$. If $\underline{r} > 0$, then $y_c^* = E w_c^* + e \in \interior\SetY$.
\end{proposition}

\begin{proof}
    We require $r G q + y_c \in \SetY$ for all $q \in \SetB(1)$. This holds if there exists a pre-image $w(q) \in \SetW$ satisfying $E w(q) + e = r G q + y_c$. Based on the concept of lifting techniques for inner approximations, we restrict this pre-image to the affine parameterization $w(q) = L q + w_c$. Matching the linear and constant terms for all $q$ yields $E L = r G$ and $y_c = E w_c + e$.

    The inclusion $w(q) \in \SetW$ requires $a_j^\top (L q + w_c) \le b_j$ for all $q \in \SetB(1)$ and all $j \in \{1, \ldots, n_c\}$. By the definition of the dual norm, this is exactly equivalent to $\|L^\top a_j\|_* + a_j^\top w_c \le b_j$. Maximizing $r$ under this restrictive parameterization yields a valid convex optimization problem that provides a lower bound $\underline{r} \le r^*$. If $\underline{r} > 0$ and $G$ is invertible, the inscribed ball is solid, ensuring that $y_c^* = E w_c^* + e \in \interior\SetY$.
\end{proof}

\subsubsection{Computing the Linear Transformation Matrix}

To mitigate point accumulation during the mapping process, \Cref{alg:basic} utilizes a direction mapping $\phi$, which for scalability is conveniently chosen as a linear transformation matrix~$\Phi$ derived from affinely related surrogate sets (cf.~\Cref{prop:ellipsoid_surrogate}). Without having access to an explicit representation of $\SetY$, existing methods to compute a surrogate set (e.g., the MVIE for polytopes) cannot be applied. Instead, we construct a solid surrogate set $\widetilde{\SetY}$ by computing an appropriate ellipsoidal surrogate in the higher-dimensional space~$\Reals{n_w}$ and mapping it forward.

To ensure that $\widetilde{\SetY}$ is solid in $\Reals{n}$, its pre-image surrogate must span the affine hull $\aff \SetW$. We achieve this by leveraging the reduced, solid set $\SetV \subset \Reals{n_v}$ introduced in the previous section. Because $\SetV$ is constructed by explicitly parameterizing the affine hull of $\SetW$ via the basis $F$ and particular solution $w_p$, any solid surrogate within $\SetV$ naturally spans this affine subspace. The following proposition details the exact recovery of the target shape matrix.

\begin{proposition}[Implicit Target Shape Matrix]\label{prop:implicit_surrogate}
    Let set $\widetilde{\SetV} = \{ Q_{\widetilde{\SetV}} u + q_{\widetilde{\SetV}} \mid \|u\|_2 \le 1 \}$ be a solid ellipsoidal surrogate for the reduced set~$\SetV$, i.e., $Q_{\widetilde{\SetV}} \in \Reals{n_v \times n_v}$ is invertible. Let the corresponding lifted surrogate set be defined as $\widetilde{\SetW} = F \widetilde{\SetV} + w_p$. If \Cref{ass:implicit} holds, the mapped surrogate in the target space, $\widetilde{\SetY} = E \widetilde{\SetW} + e$, is a solid ellipsoid. Furthermore, if $E F Q_{\widetilde{\SetV}} = U \begin{bmatrix} \Sigma & 0 \end{bmatrix} V^\top$, the target surrogate is exactly characterized by the invertible shape matrix $Q_{\widetilde{\SetY}} = U \Sigma$, meaning $\widetilde{\SetY} = \{ Q_{\widetilde{\SetY}} v + q_{\widetilde{\SetY}} \mid \|v\|_2 \le 1 \}$, with center $q_{\widetilde{\SetY}} = E(F q_{\widetilde{\SetV}} + w_p) + e$.
\end{proposition}

\begin{proof}
    Because $\widetilde{\SetV}$ is solid in $\Reals{n_v}$, it shares the affine hull of its embedding space, yielding $\aff \widetilde{\SetV} = \Reals{n_v}$. Consequently, the lifted surrogate preserves the affine hull of the base set, i.e., $\aff \widetilde{\SetW} = F(\aff \widetilde{\SetV}) + w_p = \aff \SetW$. Under \Cref{ass:implicit}, this geometric preservation guarantees that $E(\aff \widetilde{\SetW}) = \Reals{n}$.

    Mapping the lifted surrogate forward yields the target set $\widetilde{\SetY} = \{ R u + q_{\widetilde{\SetY}} \mid \|u\|_2 \le 1 \}$ with $R = E F Q_{\widetilde{\SetV}} \in \Reals{n \times n_v}$. Because $E(\aff \widetilde{\SetW}) = \Reals{n}$, the matrix $R$ is guaranteed to possess full row rank, meaning $n \le n_v$. To extract a square, invertible shape matrix $Q_{\widetilde{\SetY}} \in \Reals{n \times n}$, we apply the SVD to $R$. Substituting $R = U \begin{bmatrix} \Sigma & 0 \end{bmatrix} V^\top$ into the definition of $\widetilde{\SetY}$ reveals that the orthogonal matrix $V^\top$ acts merely as an isometry on the unit ball $\{ u \mid \|u\|_2 \le 1 \}$ in $\Reals{n_v}$, leaving its geometry completely unaltered. Subsequently, the block matrix $\begin{bmatrix} I_{n} & 0 \end{bmatrix}$ orthogonally projects this rotated $n_v$-dimensional unit ball onto its first $n$ components, which results exactly in an $n$-dimensional unit ball $\{ v \mid \|v\|_2 \le 1 \}$. Thus, the shape of the projected ellipsoid is captured by $Q_{\widetilde{\SetY}} = U \Sigma$. The matrix $Q_{\widetilde{\SetY}}$ is therefore invertible, confirming that $\widetilde{\SetY}$ is a solid ellipsoid.
\end{proof}


\begin{example}[State-Independent Transformation Matrix]\label{ex:constant_phi}
    In order to demonstrate how to compute a \emph{state-independent} linear transformation matrix $\Phi$ for the parameterized, \emph{state-dependent} linear MPC problem in \eqref{eq:ocp_terminal}, we extend the approach from \Cref{ex:ocp_interior}. Note that if the transformation matrix were state-dependent, it would require solving an optimization problem---such as an SDP to compute the MVIE---at every control cycle, thereby violating real-time constraints.

    Instead, we construct an augmented set by lifting the initial state $x$ into the input space. Let $\SetX_0 \subset \Reals{n_x}$ be a solid and convex set of viable initial states. We define the augmented set as:
    \begin{equation*}
        \SetVaug \!\triangleq\! \left\{\! (x, u_{0:N-1}) \;\middle|
        \begin{array}{l}
            x \in \SetX_0,                                                      \\
            (x_k, u_k) \in \SetZ_k, \; \forall k \in \{0, \dots, N-1\},\!\!\!\! \\
            x_N \in \SetX_N,                                                    \\
            \text{with } x_{0:N} = \bar{A}^{-1}(\bar{C} x - \bar{B} u_{0:N-1})
        \end{array}
        \right\}.
    \end{equation*}
    Provided there exists at least one initial state $x \in \interior \SetX_0$ that admits a strictly feasible trajectory (residing in the interior of the sets $\SetZ_k$ and $\SetX_N$), the joint set $\SetVaug$ contains an interior point and is therefore solid. We can therefore compute its solid ellipsoidal surrogate offline, i.e.,
    \begin{equation*}
        \widetilde{\SetV}_a = \left\{ \begin{bmatrix} Q_x \\ Q_u \end{bmatrix} \bar{v} + \begin{bmatrix} q_x \\ q_u \end{bmatrix} \;\middle|\; \|\bar{v}\|_2 \le 1 \right\}\,,
    \end{equation*}
    where the invertible block shape matrix $\begin{bmatrix} Q_x^\top & Q_u^\top \end{bmatrix}^\top$ and center $\begin{bmatrix} q_x^\top & q_u^\top \end{bmatrix}^\top$ are state independent.

    At runtime, the initial state $x$ is fixed, imposing the linear constraint $Q_x \bar{v} = x - q_x$. The general solution for the unit ball vector is
    \begin{equation*}
        \bar{v} = \bar{v}_p(x) + K_{\perp} v\,,
    \end{equation*}
    where $\bar{v}_p(x) = Q_x^\dagger (x - q_x)$ is the minimum-norm particular solution and $K_{\perp}$ is an orthonormal basis for the nullspace of $Q_x$. Because $\bar{v}_p(x)$ is orthogonal to the nullspace, we obtain $\|\bar{v}\|_2^2 = \|\bar{v}_p(x)\|_2^2 + \|v\|_2^2 \le 1$. We can therefore reparameterize the nullspace vector as $v = r(x) \tilde{v}$, where $\|\tilde{v}\|_2 \le 1$ and $r(x) = \sqrt{1 - \|\bar{v}_p(x)\|_2^2}$. Note that if $\|\bar{v}_p(x)\|_2 > 1$, the scalar $r(x)$ becomes undefined, which simply implies that the state $x$ lies outside the (inner) surrogate ellipsoid even though the true feasible set may still be nonempty. We can safely proceed without explicitly evaluating $r(x)$ because, as shown below, it acts solely as a scaling factor that is ultimately discarded.

    Substituting this parameterization back into the control sequence block yields:
    \begin{equation*}
        u_{0:N-1} = \big(r(x) Q_u K_{\perp}\big) \tilde{v} + \big(Q_u \bar{v}_p(x) + q_u\big)\,.
    \end{equation*}
    The exact shape matrix of this geometric ``slice'' is $r(x) Q_u K_{\perp}$. Crucially, $r(x)$ is only a positive scalar. Because the radial mapping in \Cref{alg:basic} inherently scales the directional transformation $\phi(d) = \Phi d$ to the boundary of the target set via the line-search parameter $\beta$, any uniform scalar factor inside $\Phi$ is completely absorbed. Thus, we can safely discard $r(x)$ and define the effective shape matrix in the control space as:
    \begin{equation*}
        Q_{\widetilde{\SetV}} = Q_u K_{\perp}\,.
    \end{equation*}
    Because $Q_u$ and $K_{\perp}$ are computed entirely offline, $Q_{\widetilde{\SetV}}$ is independent of $x$. Propagating this constant shape matrix forward via \Cref{prop:implicit_surrogate} yields a constant target shape matrix $Q_{\widetilde{\SetY}}$, and consequently, a state-independent linear transformation matrix $\Phi$.
\end{example}

\subsubsection{Computing the Boundary Scaling Factor}

Finally, to execute \Cref{alg:basic}, we must determine the boundary scaling factor $\beta > 0$ along the mapped direction $d' = \Phi d$. Rather than performing a boundary search against an explicitly computed geometric representation of $\SetY$---which would require computationally prohibitive projection operations online---we maximize the scaling factor directly by lifting the search into the higher-dimensional space of the decision variables. Because $\SetW$ is a convex set, this is formulated as a single convex optimization problem:
\begin{align}\label{eq:implicit_beta_opt}
    \max_{\beta, w} \quad & \beta                               \\
    \text{s.t.} \quad     & E w + e = y_c + \beta d', \nonumber \\
                          & w \in \SetW, \nonumber
\end{align}

By coupling the robust inscribed-ball formulations for the interior point $y_c$ (\Cref{prop:vertex_inscribed} or \Cref{prop:dual_inscribed}), the state-independent linear transformation matrix $\Phi$ derived via \Cref{prop:implicit_surrogate} and \Cref{ex:constant_phi}, and the scaling computation~\eqref{eq:implicit_beta_opt}, the mapping~$\map$ given by \Cref{alg:basic} is executed without ever computing an explicit geometric representation of the state-dependent target set $\SetY$ during the coupled execution of the RL policy and the OCP. This circumvention of complex online geometric operations ensures that FAOC remains computationally tractable and highly efficient for real-time control.

\section{Experimental Evaluation}
\label{sec:results}
\noindent
The proposed framework is deployed on a real robot that, for the first time, competes and wins against professional-level players in official ITTF matches.
Additional experimental details are provided in~\cite{ace2026}.

\begin{figure}
    \centering

    \begin{subfigure}{\linewidth}
        \centering
        \includegraphics[width=0.7\linewidth]{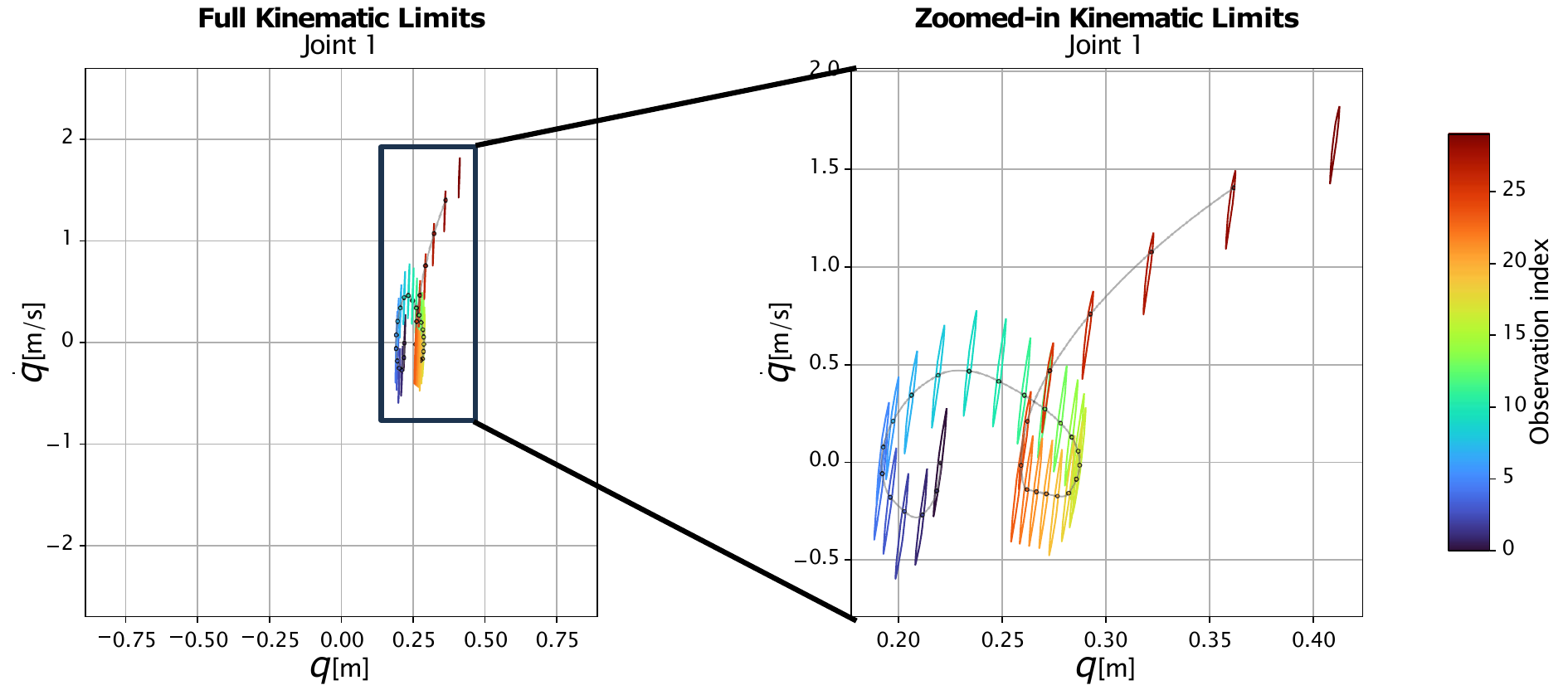} 
    \end{subfigure}
    \vspace{1cm}
    \begin{subfigure}{\linewidth}
        \centering
        \includegraphics[width=0.8\linewidth]{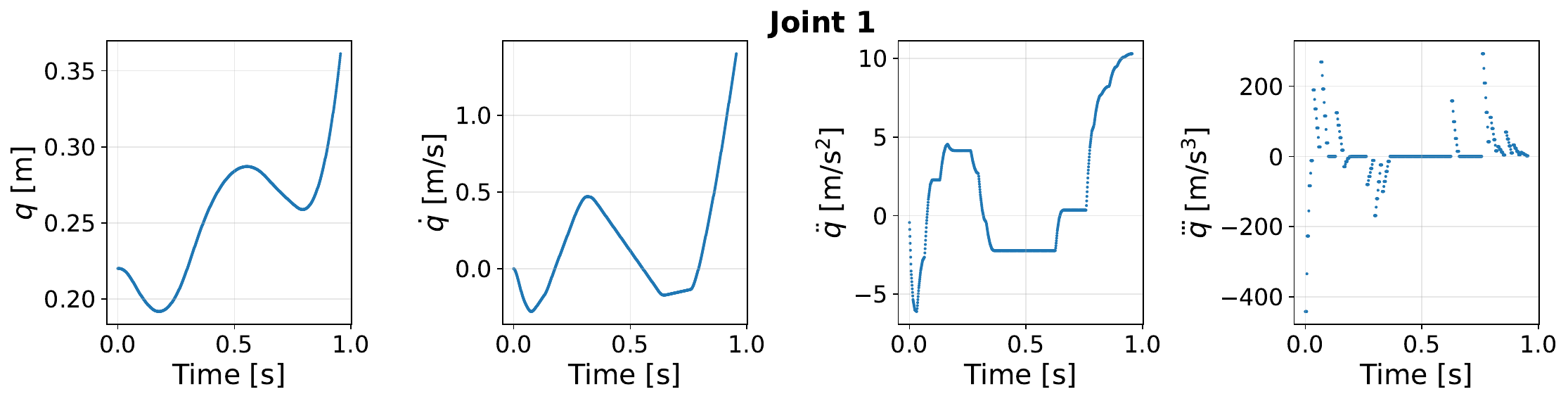} 
    \end{subfigure}

    \caption{Illustration of FAOC deployed during a table tennis game.
        Top: the executed sequence of actions (black dots), position--velocity trajectories (gray lines) and feasible polytopes for one robot joint during a robot shot.
        On the top-left, the ranges for the $x$-axis (position) and $y$-axis (velocity) correspond to the full kinematic limits of that joint.
        On the top-right, the axes are magnified for the range used during that particular motion.
        The feasible polytopes are small due to the relatively high control frequency of \qty{31.25}{\hertz}.
        Bottom: corresponding position, velocity, acceleration, and jerk trajectories as a function of time.
    }
    \label{fig:polytopes_traj}
\end{figure}

\subsection{Experimental setup}
\label{sec:exp_setup}

We first outline the hierarchical control architecture driving the robotic system. To balance strategic decision-making with hardware safety, the control stack operates across two timescales. At the high level, an RL policy processes visual and kinematic feedback to dictate the interception strategy while avoiding static obstacles (e.g., the table and robot links). At the low level, joint controllers track position references at \qty{1}{\kilo\hertz} for low latency. FAOC bridges these layers, translating the RL agent's intent into dynamically feasible joint trajectories via an OCP.

\begin{figure*}[t]
    \centering

    \begin{subfigure}[c]{0.25\textwidth}
        \centering
        \includegraphics[
            trim=0.0cm 0.0cm 2.0cm 0.0cm,
            clip,
            width=\linewidth
        ]{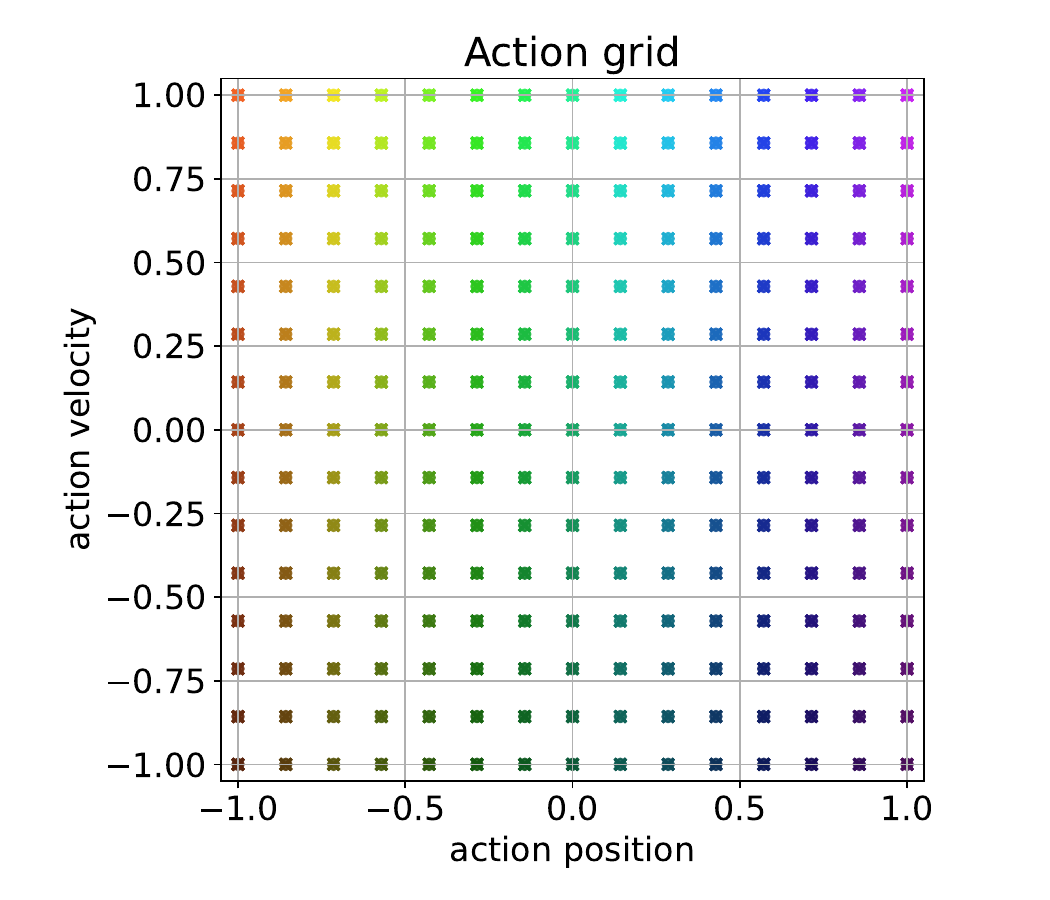}
        \caption{}
        \label{fig:action_grid}
    \end{subfigure}
    \hfill
    \begin{subfigure}[c]{0.72\textwidth}
        \centering

        \includegraphics[width=\linewidth]{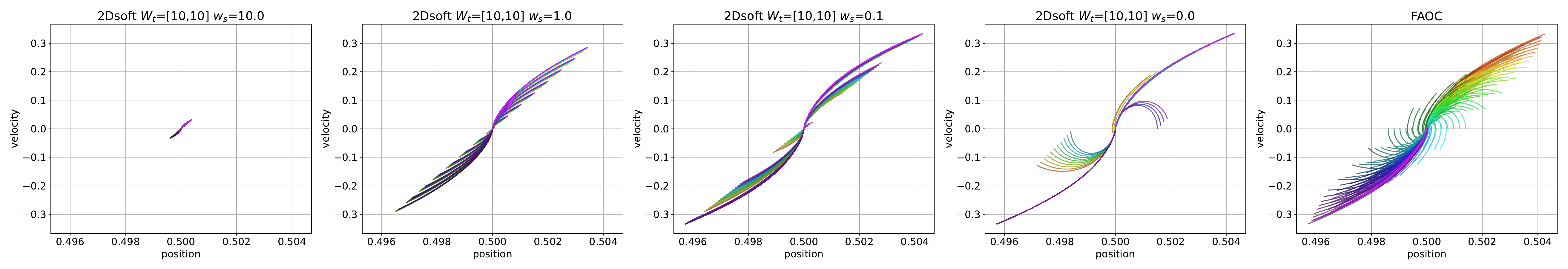}

        \vspace{0.3cm}

        \includegraphics[width=\linewidth]{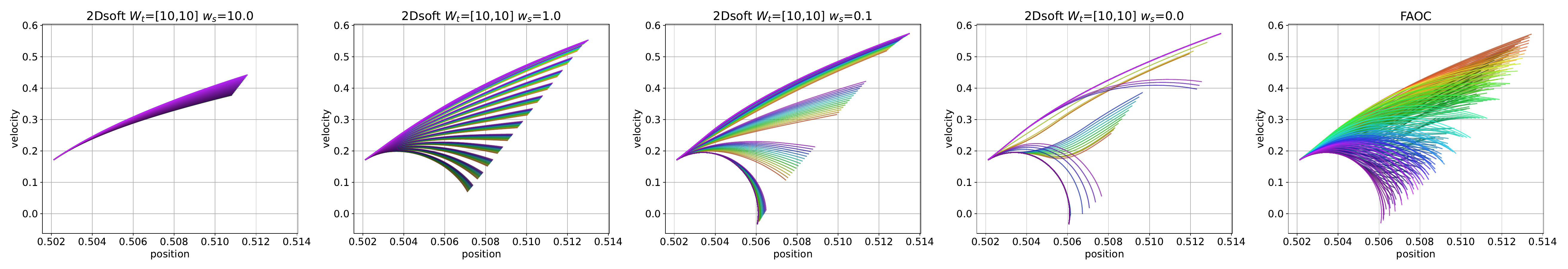}

        \caption{}
        \label{fig:joint_comparison}
    \end{subfigure}

    \caption{
        Actions sampled in a $15 \times 15$ grid for FAOC and 2Dsoft (left).
        Comparison of FAOC (right-most column) vs 2Dsoft (four left columns, for different combinations of cost weights $W_t$ and $w_s$) trajectories generated from those sampled actions for two different initial states (rows).
    }
    \label{fig:action_grid_and_comparison}
\end{figure*}

\subsubsection{Motion Planning Setup}

At each control step, the RL agent observes the full 157-dimensional system state (e.g., ball observations, joint states and desired skill) and outputs a 16-dimensional target position and velocity waypoint vector (two for each of the eight joints). These waypoints must be reached after a \qty{32}{\milli\second} horizon, therefore the RL agent acts at \qty{31.25}{\hertz}. The low-level controllers require reference positions every \qty{1}{\milli\second}, thus, for each joint, we plan a trajectory of 32 open-loop setpoints. To ensure smooth, safe motions and prevent excessive torques, these trajectories must strictly adhere to physical limits on position, velocity, acceleration, and jerk.

To meet these requirements efficiently, we solve a parameterized OCP (structured as in \Cref{cor:linear_mpc}) within the FAOC framework for each joint per RL step. Instead of planning at the native \qty{1}{\milli\second} resolution, we divide the \qty{32}{\milli\second} horizon into $N=4$ intervals of $T = \qty{8}{\milli\second}$, assuming constant jerk over each. Integrating this piecewise-constant jerk three times yields a continuous-time cubic spline for the position trajectory. The OCP is formulated as
\begin{align}\label{eq:optimization}
    \min_{x_{0:N}, u_{0:N-1}} \quad & \frac{1}{2}\sum_{k=0}^{N-1} u_k^2                                         \\
    \text{s.t.} \quad               & x_{k+1} = A x_k + B u_k, \quad \forall k \in \{0, \dots, N-1\}, \nonumber \\
                                    & x_0 = x, \nonumber                                                        \\
                                    & (x_k, u_k) \in \SetZ_k, \quad \forall k \in \{0, \dots, N-1\}, \nonumber  \\
                                    & x_{N} \in \SetX_\infty \nonumber                                          \\
                                    & C x_N = p\,, \nonumber
\end{align}
where $u_k \in \mathbb{R}$ is the jerk input and $x_k \in \mathbb{R}^3$ denotes the state (position, velocity, acceleration) at the interval boundaries. The matrices $A \in \mathbb{R}^{3 \times 3}$ and $B \in \mathbb{R}^{3 \times 1}$ represent the exact zero-order hold discretization of the continuous-time triple integrator at the $T = \qty{8}{\milli\second}$ interval:
\begin{equation*}
    A = \begin{bmatrix} 1 & T & \frac{T^2}{2} \\ 0 & 1 & T \\ 0 & 0 & 1 \end{bmatrix}, \quad B = \begin{bmatrix} \frac{T^3}{6} \\ \frac{T^2}{2} \\ T \end{bmatrix}\,.
\end{equation*}

Here, stage constraints $\SetZ_k$ encapsulate the state and input limits. To guarantee recursive feasibility, the terminal state $x_N$ is constrained to the Maximum Control Invariant Set (MCIS), denoted as $\SetX_\infty$~\cite{rawlings2017model}.

Crucially, we slightly shrink the maximum position and velocity limits. As detailed in \cite{patent2025}, this mathematically rigorous conservatism guarantees that sampling the resulting continuous-time spline at \qty{1}{\milli\second} never violates the original, unshrunk position limits or the limits on the backward-difference approximations of velocity, acceleration, and jerk. This technique decouples the OCP discretization ($T = \qty{8}{\milli\second}$) from the low-level control rate (\qty{1}{\milli\second}), reducing the OCP horizon by a factor of eight and drastically reducing the MCIS geometric complexity to enable real-time computation.

The target parameter $p \in \SetP{x} \subset \mathbb{R}^2$ defines the desired terminal waypoint. It is generated by mapping the abstract RL action $\bar{a}$ using the techniques from \Cref{ssec:implicit}. The full-row-rank matrix $C = \begin{bmatrix} I_{2} & 0 \end{bmatrix} \in \mathbb{R}^{2 \times 3}$ isolates the position and velocity components from $x_N$ to enforce the target.

The resulting quadratic programs are solved in realtime with DAQP~\cite{arnstrom2022dual}. The MCIS is computed offline using Qhull~\cite{barber96qhull} to eliminate redundancies. While the 2D area-matching transformation is applicable, we exclusively use the linear transformation due to its lower computational cost, superior numerical robustness, and equivalent performance.

To confirm real-time viability within the \qty{32}{\milli\second} control period (\qty{31.25}{\hertz}), execution latencies were benchmarked across \qty{25000}{} control cycles on an Intel Core Ultra 9 285K CPU. The computational breakdown per control cycle for one joint is as follows:
\begin{itemize}
    \item \textbf{Mapping Computation:} Computing the interior point~$y_c$ and evaluating the boundary scaling factor $\beta$ via \eqref{eq:implicit_beta_opt} requires an average of \qty{22.54}{\micro\second} (worst-case \qty{82.97}{\micro\second}). 
    \item \textbf{OCP Solve Time:} Solving the quadratic program with DAQP requires an average of \qty{2.49}{\micro\second} (worst-case \qty{21.95}{\micro\second}).
    \item Other operations (e.g. spline sampling), accounting for less than \qty{1}{\micro\second} of computations.
\end{itemize}
The above results in an average total control-cycle computation time of \qty{25.23}{\micro\second} with a maximum peak latency of \qty{91.57}{\micro\second}, per joint.
If this sequence is run in series for all $8$ joints, the expected average latency is then approximately \qty{200}{\micro\second}, but using a 4-core parallelization strategy, this can be reduced to \qty{72.2}{\micro\second}, or \qty{465.7}{\micro\second} in the worst observed case, utilizing thus less than $1.5\%$ of the available \qty{32}{\milli\second} real-time deadline even when assuming the worst-case scenario.
Empirical testing also revealed that pinning optimization processes to isolated CPU cores (preventing OS-level background task contention) significantly reduced maximum observed runtimes.
    
\subsubsection{Compared controllers}

First, we implement the method of~\cite{kiemel2021learning} using the \texttt{klimits} library\footnote{\url{https://github.com/translearn/limits}, version 1.1.3, MIT License}. Like FAOC, this method incorporates reachability into the action space design. However, it is limited to one-dimensional waypoints, so the RL agent specifies either position, velocity, or acceleration. We refer to these variants as \textit{1Dpos}, \textit{1Dvel}, and \textit{1Dacc}.

We also compare FAOC with a two-dimensional controller that does not account for reachability when defining the action space.
The RL policy outputs a position--velocity waypoint over the full kinematic range (cf. top-left rectangle in \Cref{fig:polytopes_traj}), which may therefore be infeasible. To accommodate infeasible targets, the OCP replaces the terminal equality constraint in~\eqref{eq:optimization} with the soft terminal cost
\begin{equation*}
    \frac{1}{2}\left(
    \sum_{k=0}^{N-1} w_s u_k^2
    +
    \|p-Cx_N\|_{W_t}^2
    \right)\,,
\end{equation*}
where $w_s \in \Reals{}_{++}$ and $W_t \in \mathbb{S}_{++}^2$ are the stage and terminal weights, respectively. This formulation is implemented with the \textit{Clarabel} solver through the \texttt{qpsolvers} library~\cite{qpsolvers}\footnote{\url{https://github.com/qpsolvers/qpsolvers}, version 3.4.0, LGPL-3.0 license.}, and is referred to as \textit{2Dsoft}.

Unlike FAOC, 2Dsoft uses a static, state-independent action space and relies on the cost weights to balance trajectory smoothness against waypoint tracking. As a result, the action space may contain unreachable waypoints, introducing redundant actions whose extent depends on the robot state and the choice of $w_s$ and $W_t$. In contrast, FAOC guarantees feasible waypoints through a terminal equality constraint and requires no tuning. These differences are illustrated in \Cref{fig:action_grid_and_comparison}.

\subsubsection{Reinforcement Learning Problem Formulation}
We formulate the RL agent using Universal Value Function Approximators (UVFAs)~\cite{schaul2015universal} to learn three table-tennis skills: ball placement, topspin/backspin generation, and smashes.
At the start of each training episode, one skill is sampled uniformly at random.
Each episode consists of returning a single incoming shot and terminates upon a successful return or a miss.
Training shots are sampled from a mixture of synthetic trajectories and real-world human shots to increase shot diversity.
The agent receives a reward only at the end of each episode, resulting in a sparse-reward problem.
Further details are provided in~\cite{ace2026}.

\subsubsection{Evaluation metrics}
The evaluation metrics quantify how effectively the agents learn the target skills. Each data point in \Cref{fig:stat_analysis,fig:training_curves} is computed from 1000 evaluation episodes.
\textit{Fraction Balls Returned} measures the return rate, while \textit{Generated Topspin}, \textit{Generated Velocity}, \textit{Generated Backspin}, and \textit{Goal Distance}, evaluate topspin, return velocity, backspin, and target accuracy, respectively.
Higher values indicate better performance except for \textit{Generated Backspin} and \textit{Goal Distance}, as indicated by the arrows in \Cref{fig:stat_analysis}.

All metrics are computed from the ball state at its bounce on the opponent's side and are therefore undefined for unsuccessful returns.
To reduce estimation variance, we exclude evaluations with a return rate below 50\%.
For each of the five controllers (FAOC, 2Dsoft, 1Dpos, 1Dvel, and 1Dacc), we train five agents using different random seeds.

\begin{figure}[t!]
    \centering
    \includegraphics[width=0.7\linewidth]{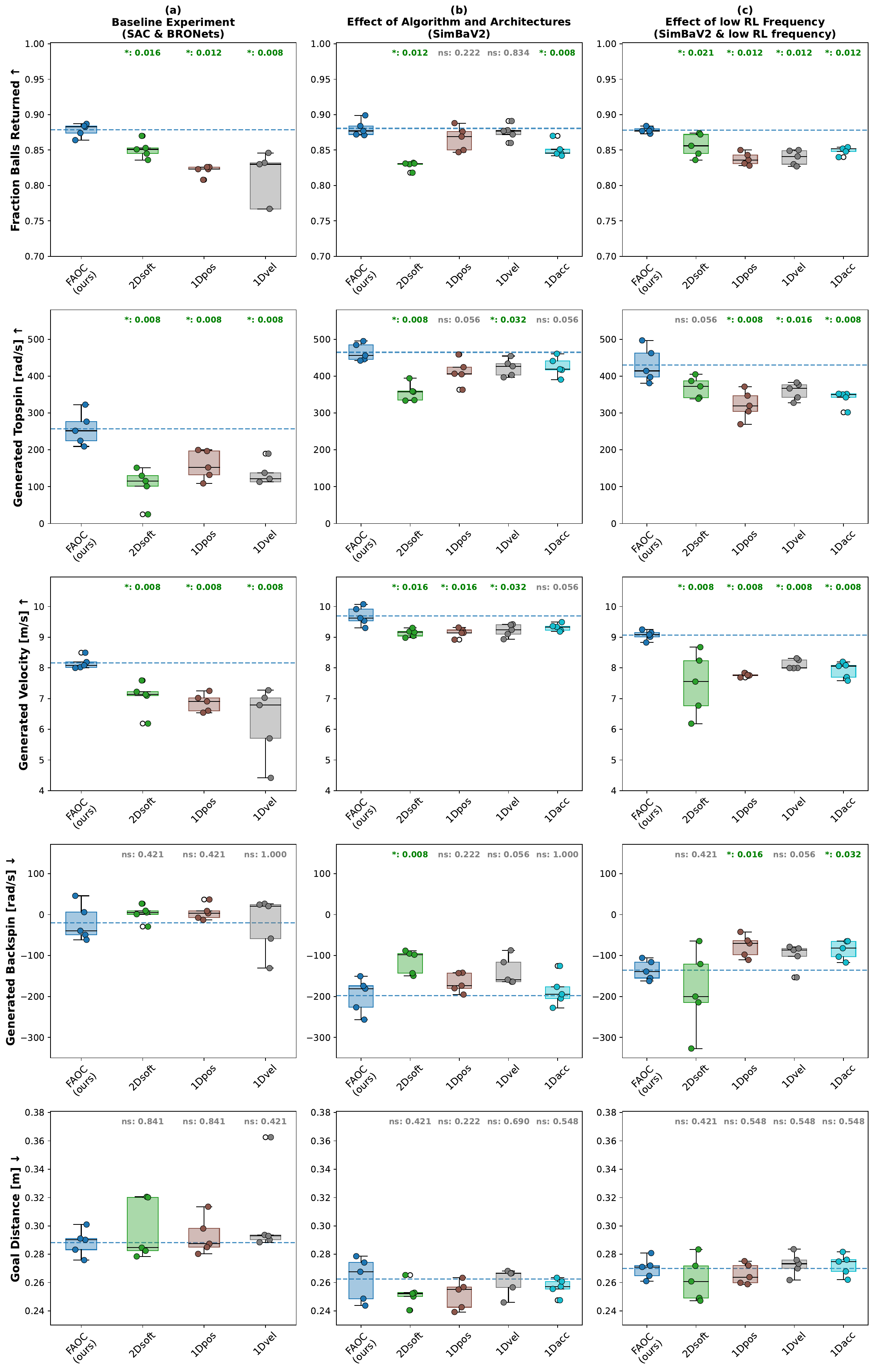}
    \caption{Statistical comparison of performance metrics (plot rows) for the five methods (five seeds each), i.e., FAOC (ours), 2Dsoft, 1Dpos, 1Dvel, and 1Dacc, for the three conducted experiments (plot columns), i.e., (a) baseline experiment using SAC, (b) study of the effect of algorithm and architectures using SimBaV2, and (c) study of the effect of lower RL frequency with SimBaV2. Dashed lines indicate mean value of FAOC (ours). For the baseline experiment, results for 1Dacc are not reported since three seeds did not learn to return balls to the opponent. Arrows next to metrics names indicate which direction is better for the metric, either higher is better (up arrow) or lower is better (down arrow). Above each boxplot of the methods we compare FAOC against, we report p-values of pairwise Mann-Whitney U tests between the corresponding method and FAOC, where $*$ denotes a statistically significant advantage for FAOC ($p<0.05$) while $ns$ indicates no significant difference. Importantly, none of the methods is significantly better than FAOC in any of the metrics for any of the experiments.}
    \label{fig:stat_analysis}
    \vspace{-0.5cm}
\end{figure}

\subsection{Baseline experiment (SAC + BRONets)}
\label{sec:SAC-comparison}
We compare FAOC, 2Dsoft, 1Dpos, 1Dvel, and 1Dacc using off-policy RL algorithm Soft Actor-Critic (SAC)~\cite{haarnoja2018soft} with BRO-based policy and critic network architectures~\cite{nauman2024bigger}. 
Hyperparameters and additional training details are provided in~\cite{ace2026}. 
\Cref{fig:sac-comparison} reports normalized reward as a function of environment steps (state--action--reward--next-state transitions).
FAOC achieves the highest sample efficiency and final performance, while the two 2D controllers (FAOC and 2Dsoft) consistently outperform the 1D variants. The latter also exhibit substantially larger variability across random seeds. Statistical comparisons of the evaluation metrics are shown in \Cref{fig:stat_analysis}a using the best-performing checkpoint from each of the five training seeds. The 1Dacc controller is excluded because three of its five seeds achieve a return rate below 50\%.

\begin{figure*}[t]
    \centering
    \begin{subfigure}{0.37\linewidth}
        \centering
        \includegraphics[width=\linewidth]{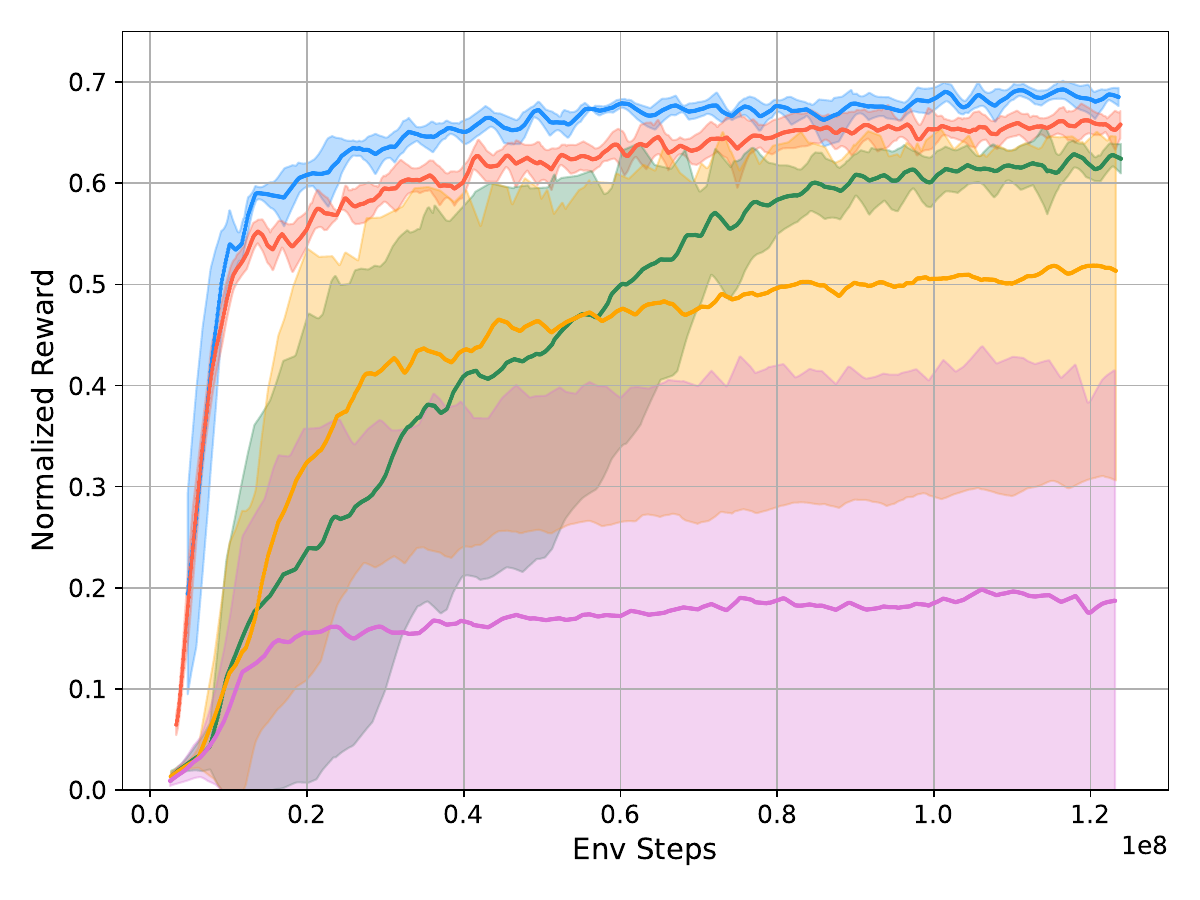}
        \captionsetup{justification=centering}
        \caption{Baseline experiment \\(SAC \& BRONets)}
        \label{fig:sac-comparison}
    \end{subfigure}
    \hfill
    \begin{subfigure}{0.305\linewidth}
        \centering
        \includegraphics[width=\linewidth]{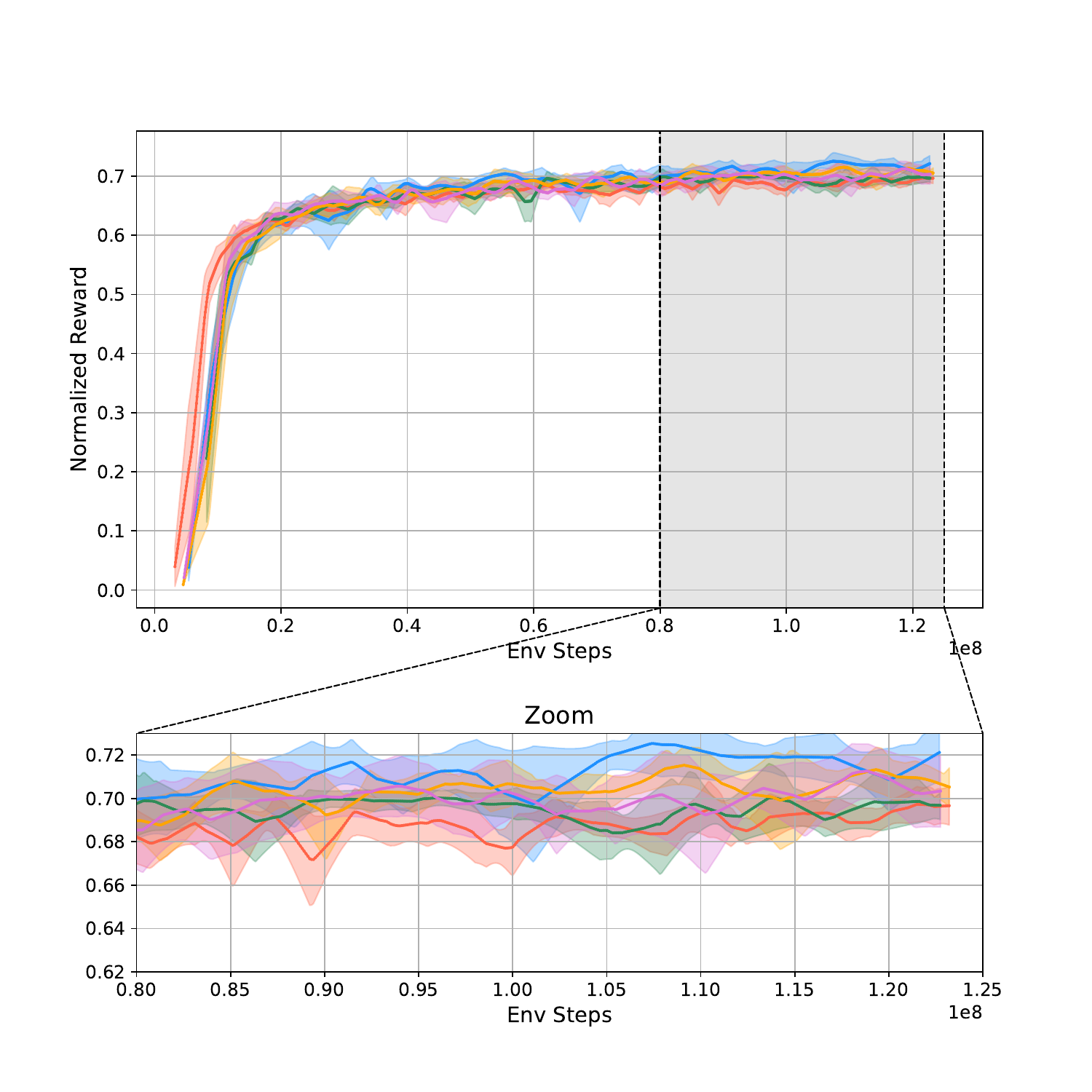}
        \captionsetup{justification=centering}
        \caption{Effect of algorithm and architecture (SimBaV2)}
        \label{fig:simbav2-comparison}
    \end{subfigure}
    \hfill
    \begin{subfigure}{0.305\linewidth}
        \centering
        \includegraphics[width=\linewidth]{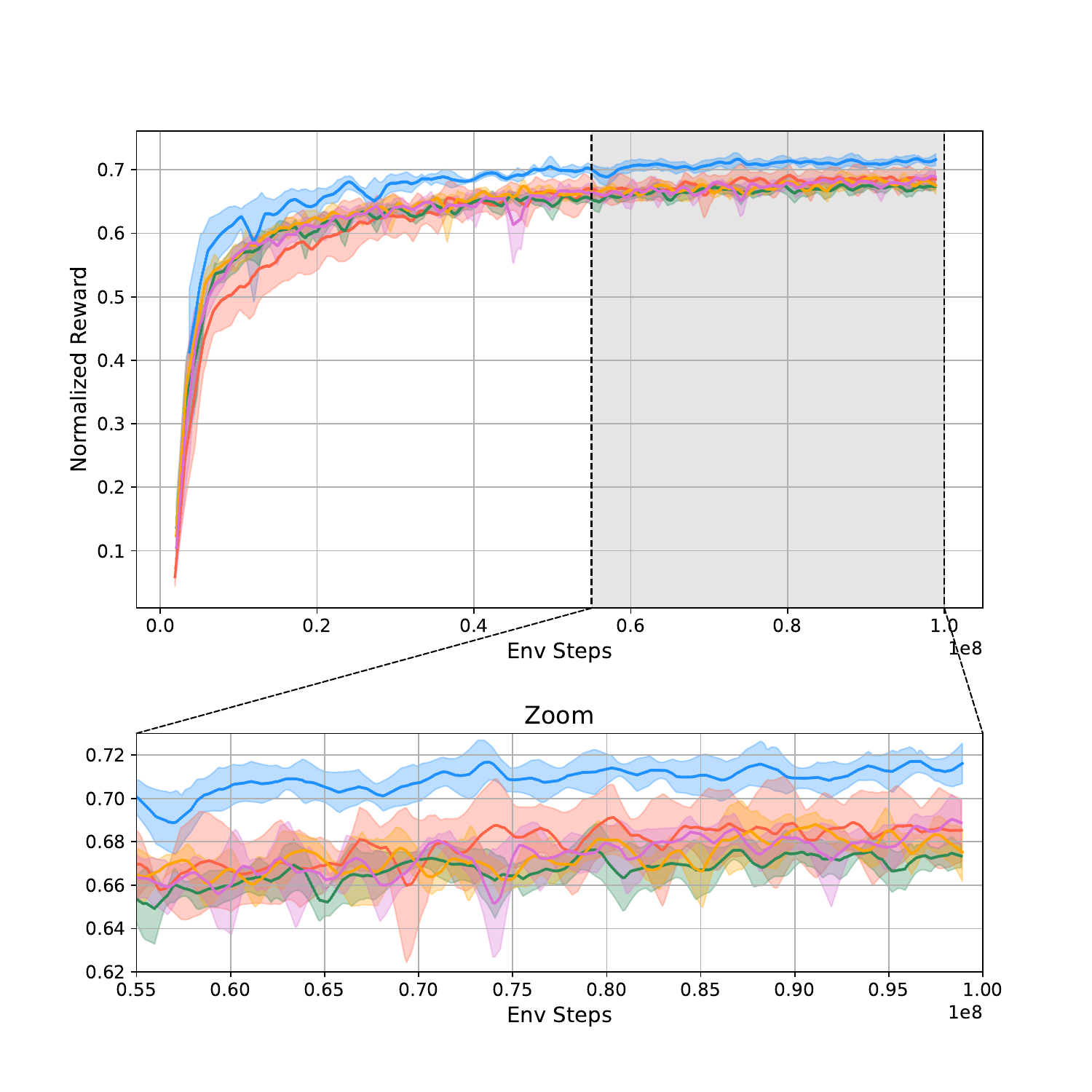}
        \captionsetup{justification=centering}
        \caption{Effect of low RL frequency \\(SimBaV2 \& low RL frequency)}
        \label{fig:128ms-comparison}
    \end{subfigure}
    \hfill
    \begin{subfigure}{0.49\textwidth}
        \centering
        \begin{tikzpicture}
            \node[draw=none]{
                \vspace{4pt}

                \makebox[0.9\textwidth][c]{
                    \small
                    \begin{tabular}{ccccc}
                        \legendline{dodgerblue}{FAOC (ours)} &
                        \legendline{tomato}{2Dsoft}          &
                        \legendline{seagreen}{1Dpos}         &
                        \legendline{orange}{1Dvel}           &
                        \legendline{orchid}{1Dacc}
                    \end{tabular}
                }
            };
        \end{tikzpicture}
    \end{subfigure}

    \caption{Means (bold lines) and standard deviations (shaded areas) of normalized reward as a function of environment steps across five seeds for each controller. (a) Baseline experiment using SAC algorithm and BRONet architectures. (b) Effect of algorithm and architectures using SimBaV2. (c) Effect of low RL frequency using SimBaV2. Experiment (b) tests whether FAOC is intrinsically better than the other methods, regardless of algorithm and architectures. Experiment (c) tests the hypothesis that FAOC features better trajectory segment controllability by combining 2D waypoints and reachability awareness through its state-dependent action space formulation. Comparing zoomed-in areas of (b) and (c) reveals that FAOC retains better controllability than the other controllers when reducing RL frequency.}
    \label{fig:training_curves}
\end{figure*}

FAOC significantly outperforms all other methods in return rate, generated topspin, and generated velocity, explaining its superior reward.
No statistically significant differences are observed in generated backspin or goal distance.
All methods generate little backspin, likely because cutting the ball at high speed while avoiding the table is challenging.
Likewise, placing the ball at any arbitrary location regardless of the incoming shot is difficult.
The measured aiming error (approximately \qty{30}{\centi\meter}) is sufficient for coarse shot placement and is consistent with the conservative reward design used to improve zero-shot transfer to the real robot.

The inferior performance of the 1D controllers can be attributed to their reduced controllability: the RL agent directly specifies only one waypoint variable, leaving the remaining degrees of freedom to the optimizer. This limits the set of achievable behaviors and increases sensitivity to exploration, resulting in lower performance and higher variance across seeds. Although 2Dsoft provides the same controllability as FAOC, its state-independent action space contains unreachable waypoints. Exploring such infeasible actions reduces sample efficiency, making it more difficult to discover optimal policies than with FAOC.

\subsection{Effect of algorithm and architectures (SimBaV2)}
\label{ssec:simbav2-comparison}

To assess whether FAOC's advantage persists with stronger RL methods, we repeat the comparison using the algorithmic and architectural modifications of SimBaV2~\cite{lee2025hyperspherical}. Specifically, we replace the critics with categorical distributional critics, which have been shown to improve exploration and RL performance~\cite{sun2025intrinsic,farebrother2024stop}. Additionally, we increase the actor and critic sizes by factors of $1.71$ (i.e., from 2,570,073 to 4,407,385 parameters) and $8.07$ (i.e., from 6,380,545 to 51,519,520 parameters), respectively, relative to \Cref{sec:SAC-comparison}. This has been shown to improve expressivity, resulting in better exploration~\cite{wang20251000}. Training curves and statistical comparisons are reported in \Cref{fig:simbav2-comparison} and \Cref{fig:stat_analysis}b.



Because the rewards for spin and velocity are proportional to $\tanh(\text{spin}/450)$ and $\tanh(\text{vel}/8)$, improvements at large values in these metrics produce only small reward differences.
Consequently, the normalized reward curves appear similar despite the clearer distinctions observed in \Cref{fig:stat_analysis}b.
All controllers improve substantially over the baseline, yet FAOC consistently remains among the best-performing methods and retains statistically significant superiority in eight comparisons. These results indicate that although stronger RL algorithms and larger networks improve all controllers, FAOC's action space formulation continues to provide a performance advantage.

\subsection{Effect of low RL frequency (SimBaV2 \& low RL frequency)}
\label{sec:RL-freq}

The experiments above use an RL decision frequency of \qty{31.25}{\hertz}, allowing the policy to replan every \qty{32}{\milli\second}. Frequent replanning improves controllability because each action shapes only a short trajectory segment, reducing the impact of suboptimal decisions. However, it also increases the number of decisions per episode, making credit assignment more difficult.
We want to consider systems where, e.g., due to mechanical, sensory or communication latency, only lower control frequency is possible.
Visual-motor latency of humans can be as fast as \qty{150}{\milli\second}, so for this experiment we choose trajectory segments of \qty{128}{\milli\second}, corresponding to an RL frequency of \qty{7.8125}{\hertz}.
All controllers are trained using the same algorithm and architectures as in \Cref{ssec:simbav2-comparison}.

\Cref{fig:128ms-comparison} shows that FAOC largely preserves the performance achieved at the higher decision frequency, as highlighted by the zoomed view of the final 45 million environment steps. In contrast, the remaining controllers degrade more substantially. This trend is confirmed by the statistical analysis in \Cref{fig:stat_analysis}c, where FAOC is significantly better than the competing methods in 13 pairwise comparisons, compared with 8 in \Cref{fig:stat_analysis}b.

\section{Conclusion}
\label{sec:conclusions}
\noindent
We introduced Feasible Action for Optimal Control (FAOC), a control framework that maps Reinforcement Learning (RL) actions from a static convex set to a state-dependent feasible parameter set. We demonstrated FAOC on robot table tennis, where it enabled a physical robot to compete and win against professional-level players in official matches~\cite{ace2026}. 
Experimental results confirm that FAOC consistently outperforms existing RL-OC architectures by seamlessly combining strict reachability awareness with superior trajectory segment controllability.

\section*{Acknowledgment}
\noindent
The authors sincerely thank Daniel Arnström for his prompt, invaluable DAQP support throughout this project. 
Additionally, the authors acknowledge Christian Conti, Dunai Fuentes Hitos, Mario Ynocente Castro, Lison Abecassis, Yannik Nagel, Andrea Scotti, Etienne Walther, Jengyan Wong, and the platform team of Sony AI led by Piyush Khandelwal for their contributions to the training pipeline used in the experiments of this paper.

\bibliographystyle{IEEEtran}
\bibliography{biblio}

\end{document}